\documentclass[11pt]{article}
\usepackage[latin9]{inputenc}
\usepackage{float}
\usepackage{bm}
\usepackage{amsmath}
\usepackage{amsthm}
\usepackage{graphicx}
\usepackage{authblk}

\makeatletter

\providecommand{\tabularnewline}{\\}

\theoremstyle{plain}
\newtheorem{thm}{\protect\theoremname}
\theoremstyle{definition}
\newtheorem{defn}[thm]{\protect\definitionname}
\theoremstyle{plain}
\newtheorem{prop}[thm]{\protect\propositionname}
\theoremstyle{remark}
\newtheorem{rem}[thm]{\protect\remarkname}

\usepackage{amsmath,amssymb,amsfonts,amsthm}
\usepackage[normalem]{ulem}
\usepackage{soul}
\usepackage{moreverb,rotating,graphics}
\usepackage[round]{natbib}
\usepackage{hyperref}
\hypersetup{
	colorlinks,
	linkcolor={blue!50!blue},
	citecolor={red},
	urlcolor={blue!50!blue}
}
\usepackage{float}
\usepackage{appendix}
\usepackage{epsfig}
\usepackage{tikz}
\usepackage{graphicx}
\usepackage{lipsum}
\usepackage{enumerate}
\usepackage{import}
\usepackage[normalem]{ulem}
\usepackage{color}
\usepackage{caption}
\usepackage{subfigure}
\usepackage[ruled,vlined]{algorithm2e}
\usepackage{bm}
\usepackage[absolute,overlay]{textpos}
\usepackage{array}
\usepackage{multicol}
\usepackage{mathrsfs}
\usepackage{tablefootnote}
\usepackage{minitoc}
\usepackage{cases}
\allowdisplaybreaks

\usepackage[Conny]{fncychap}
\usepackage[a4paper]{geometry}
\def\sqw{\hbox{\rlap{\leavevmode\raise.3ex\hbox{$\sqcap$}}$%
\sqcup$}}
\def\cqfd{\ifmmode\sqw\else{\ifhmode\unskip\fi\nobreak\hfil
\penalty50\hskip1em\null\nobreak\hfil\sqw
\parfillskip=0pt\finalhyphendemerits=0\endgraf}\fi}

\renewcommand{\eqref}[1]{(\ref{#1})}

\usepackage[normalem]{ulem}
\usepackage{todonotes}

\newcommand\MSE{\mathrm{MSE}}

\newcommand{\dd}{\mathrm d}

\newcommand{\cC}{\mathcal{C}}
\newcommand{\cD}{\mathcal{D}}

\newcommand{\cG}{\mathcal{G}}

\newcommand{\cN}{\mathcal{N}}
\newcommand{\cO}{\mathcal{O}}

\newcommand{\cR}{\mathcal{R}}

\newcommand{\cT}{\mathcal{T}}

\newcommand{\bE}{{\mathbb E}}

\newcommand{\bN}{{\mathbb N}}

\newcommand{\esp}{\mathbb E}

\newcommand{\be}{\begin{equation}}
\newcommand{\ee}{\end{equation}}

\newcommand{\maj}{>}
\newcommand{\mino}{<}

\newcommand{\Rn}{\VIX_{T}^{2,\cR_{n}}}

\newcommand{\Tn}{\VIX_{T}^{2,\cT_{n}}}
\newcommand{\Dn}{\VIX_{T}^{2,\cD_{n}}}
\newcommand{\Cost}{\mathrm{Cost}}
\newcommand{\CostML}{\mathrm{Cost}_{\mathrm{ML}}}
\newcommand{\CostMC}{\mathrm{Cost}_{\mathrm{MC}}}

\newcommand{\pricemkt}{\mathrm{price}^{\mathrm{mkt}}}
\newcommand{\PML}{\widehat{P}_{\bm{M},\bm{n}}}

\newcommand{\Lip}{\mathrm{L}_{\varphi}}

\newcommand{\VIX}{\mathrm{VIX}}
\newcommand{\ve}{\varepsilon}

\newcommand{\BS}{\mathrm{BS}}

\newcommand{\SP}{\rm{S\&P}500}

\newcommand{\egl}{\overset{d}{=}}

\DeclareMathAlphabet\mathbfcal{OMS}{cmsy}{b}{n}

\newcommand{\Var}{\mathrm {Var}}
\newcommand{\Cov}{\mathbb{C}\mathrm{ov}}

\makeatother

\providecommand{\definitionname}{Definition}
\providecommand{\propositionname}{Proposition}
\providecommand{\remarkname}{Remark}
\providecommand{\theoremname}{Theorem}

\begin{document}

\title{Multilevel Monte Carlo simulation for $\VIX$ options in the rough Bergomi model}
\author[1,2]{F.\ Bourgey\thanks{(Corresponding author) florian.bourgey@polytechnique.edu, bourgeyflorian@gmail.com}}
\author[1]{S.\ De Marco\thanks{stefano.de-marco@polytechnique.edu}}
\affil[1]{Centre de Math\'ematiques Appliqu\'ees (CMAP), CNRS, Ecole Polytechnique, Institut Polytechnique de Paris. Route de Saclay, 91128 Palaiseau Cedex, France. }
\affil[2]{Bloomberg L.P., Quantitative Research, 3 Queen Victoria St, London
	EC4N 4TQ, UK.}

\date{}
\maketitle

{

}

\begin{abstract}
	\noindent We consider the pricing of $\VIX$ options in the rough Bergomi model. %\citep{bayer2016pricing}
	%[Bayer, Friz, and Gatheral, \emph{Pricing under rough volatility}, Quantitative Finance 16(6), 887--904, 2016]
	In this setting, the $\VIX$ random variable is defined by the one-dimensional integral of the exponential of a Gaussian process with correlated increments, hence
	approximate samples of the $\VIX$ can be constructed via discretization of the integral and simulation of a correlated Gaussian vector.
	A Monte-Carlo estimator of $\VIX$ options based on a rectangle discretization scheme and exact Gaussian sampling via the Cholesky method has a computational complexity of order $\cO(\ve^{-4})$ when the mean-squared error is set to $\ve^2$.
	We demonstrate that this cost can be reduced
	to $\cO(\ve^{-2} \log^2(\ve))$ combining the scheme above with the multilevel method,
	and further reduced to the asymptotically optimal cost $\cO(\ve^{-2})$ when using a trapezoidal discretization.
	We provide numerical experiments highlighting the efficiency of the multilevel approach in the pricing of $\VIX$ options in such a rough forward variance setting.
\end{abstract}

\medskip

\section{Introduction }

The joint modeling of the dynamics of option prices  and of the underlying financial asset $S$ has been one of the central problems in mathematical finance over the last decades.
An explicit modeling of the implied volatility surface (which is an infinite-dimensional object parameterized by two variables, the option maturity $T$ and the strike price $K$) being a difficult problem, a lot of attention has turned towards the modeling of forward variances, \emph{aka} the implied variances of forward log-contracts.
The forward variance $V_t^{T, T + \Delta}$ observed at time $t$ for the future period $[T, T + \Delta]$ is defined by the market price prevailing at time $t$ of a forward log-contract,
%over $[T_1, T_2]$
$V_t^{T, T + \Delta} := \pricemkt_t \Bigl(-\frac 2 {\Delta} \log \Bigl( \frac{S_{T + \Delta}}{S_{T}} \Bigr) \Bigr)$. Forward variances are still infinite-dimensional objects, but one is now focusing on a \emph{one-dimensional} set of derivatives (as opposed to the two-dimensional implied volatility surface), parameterized by the future maturity $T$, and the modeling problem is considerably simplified.
Note that, taking into account the prefactor $- \frac 2{\Delta}$ in the log-contract payoff, the Black--Scholes implied variance of the log-contract coincides with the derivative's price: the price of such a contract in a Black--Scholes model with volatility parameter $\sigma$ is indeed $\mathrm{price}_t^{\mathrm{BS}} \Bigl(-\frac 2 {\Delta} \log \Bigl( \frac{S_{T+\Delta}}{S_{T}} \Bigr) \Bigr) = \sigma^2$.

The implied variance of the log-contract is used by the Chicago Board Options Exchange (CBOE) to define the $\VIX$ index \citep{exchange2009cboe}: more precisely, the VIX value at time $T$ is equal to the implied volatility of the log-contract delivering $-\frac 2 {\Delta}\ln\left(\frac{S_{T+\Delta}}{S_{T}}\right)$ at time $T+\Delta$, where $\Delta=30$ days and $S$ is the $\SP$ index:
\begin{equation}\label{eq:def:vix:impliedvoloflogcontract}
	\VIX_{T} :=\sqrt{\pricemkt_{T}\biggl(-\frac{2}{\Delta}\ln\biggl(\frac{S_{T+\Delta}}{S_{T}}\biggr)\biggr)}.%=\sqrt{\price_{T}\left(\frac{1}{\Delta}\int_{T}^{T+\Delta}\sigma_{\alpha}^{2}\dd\alpha\right)}.
\end{equation}
For simplicity, we consider zero interest, dividend, and repo rates here and in what follows. If this is not the case, $S_T$ should be replaced inside \eqref{eq:def:vix:impliedvoloflogcontract} by the $\SP$ forward price $F_T^{T+\Delta}$ observed at time $T$ for the maturity $T+\Delta$.
Since the log-contract itself is not  a traded object on the CBOE,  its price is quoted  by static replication of the payoff $\ln(S)$ with call and put options; the precise way the static replication formula is discretized over the observed option strikes and maturities is described in the VIX White paper \citep{exchange2009cboe}.

A standard and widespread approach, pioneered by \citep{dupire1992arbitrage} and intensively studied by \citep{bergomi2005smile, Bergomi2016}, is to set up a model for instantaneous forward variances, defined by $\xi_t^u = \frac{\dd}{\dd u} \bigl( (u - t)  V_t^{t, u} \bigr)$ for $u \ge t$  (very much in the spirit of the definition of instantaneous forward rates from bond prices).
Once a model for the instantaneous -- and therefore unobserved -- object $\xi_t^u$ has been settled, the discrete forward variances $V_t^{T, T + \Delta}$ are retrieved, for any tenor structure $T$, via integration over the parameter $u$, $V_t^{T, T + \Delta} = \frac 1 \Delta \int_T^{T+\Delta} \xi_t^u \dd u$.

In a class of models widely used in practice since the works of \citep{bergomi2005smile}, forward variances are modeled by the exponential
of a Gaussian process
\begin{equation} \label{def:sde:X_T}
X_{t}^{u}
:= \log \left( \xi_t^u \right)
= X_{0}^{u} - \frac{1}{2}\int_{0}^{t} K\left(u,s\right)^{2}\dd s+\int_{0}^{t} K\left(u,s\right)\dd W_{s},\quad u\geq t,
\end{equation}
where $W$ is a standard Brownian motion and $K$ a deterministic kernel such that $K(u, \cdot) \in L^2((0,u))$ for every $u \maj 0$.
The class of models \eqref{def:sde:X_T} encompasses Bergomi's model \citep{bergomi2005smile}, obtained when choosing the exponential kernel $K(u,s) =  \omega \, e^{-\kappa(u - s)}$ with
$\omega, \kappa \maj 0$, and the so-called rough Bergomi model of \citep{bayer2016pricing}, corresponding to the power-law (or fractional) kernel $K\left(u,s\right)=  \frac \eta{\left(u-s\right)^{1/2 - H}}$, where $H\in\left(0,1/2\right)$ and $\eta>0$. 
Several authors have reported that a small value of $H$, namely $H \approx 0.1$, is appropriate to reconstruct the observed term structure both of the at-the-money implied volatility skew of the $\SP$ \citep{alos2007short, bayer2016pricing, fukasawa:2017} and of VIX futures and at-the-money VIX implied volatilities \citep{Jacquier2018VIXFutures, Alos2018VIXskew}.

The modeling framework above has an appealing simplicity: being explicit functions of a Gaussian process, instantaneous forward variances $\xi_t^u = e^{X_t^u}$ can be simulated exactly.
Nevertheless, this modeling approach rises some non-trivial computational issues: recall that the instantaneous forward variance $\xi_t^u$ is an unobservable object, while traded derivatives are written on integrated forward variances. More specifically, the squared $\VIX$ index at time $T$ is
\begin{equation}\label{eq:def:VIX2}
(\VIX_{T})^2 = \frac{1}{\Delta}\int_{T}^{T+\Delta}e^{X_{T}^{u}}\dd u \,,
\end{equation}
so that the price at time zero of an option on the $\VIX^2$ with payoff $\varphi$ is given by
\begin{equation}\label{def:option:VIX2}
	\bE\left[\varphi\left(\VIX_{T}^{2}\right)\right]=\bE\left[\varphi\left(\frac{1}{\Delta}\int_{T}^{T+\Delta}e^{X_{T}^{u}}\dd u\right)\right].
\end{equation}
When $\varphi\left(x\right)=\left(\sqrt{x}-\kappa\right)_{+}$ (resp.\ $\varphi\left(x\right)=\left(\kappa-\sqrt{x}\right)_{+}$), \eqref{def:option:VIX2} is the price of a VIX call option (resp.\ put option) with strike $\kappa$ and maturity $T$.
In any forward variance model such as \eqref{def:sde:X_T}, then, the $\VIX$ is given by a continuous sum of log-normal random variables; the resulting distribution is unknown (apart from trivial cases where $K \equiv 0$) and the random variable $\VIX_T$ cannot be simulated exactly. 

The problem \eqref{def:option:VIX2} is close in spirit to the valuation of Asian options $\varphi(A_T) = \varphi\left(\frac{1}{T}\int_{0}^{T} Y_{t} \, \dd t\right)$, where $Y$ represents an asset price.
There is a vast literature on the numerical methods that can be used to evaluate $\esp[\varphi(A_T)]$ when $Y$ is modeled by a diffusion process.
Yet, there is a fundamental and structural difference with the pricing of a VIX option such as \eqref{def:option:VIX2}: while the Asian option payoff is based on the integral of a continuous--time Markov process with respect to its time parameter, there is, in general, no finite--dimensional Markovian structure underlying the problem \eqref{def:option:VIX2}, and this is precisely the case when $K$ is a fractional kernel,
%$K\left(u,s\right)=  \eta \left(u-s\right)^{H -1/2}$,
as in the rough Bergomi model.
In particular then, numerical approaches based on PDEs, explored in the context of Asian options by \citep{rogers1995value,zhang2001semi}, and \citep{dubois2004efficient}, are arguably out-of-the-way (in a setting of finite dimensional PDEs at least).

The fact that the random vector $\bigl(e^{X_T^{u_i}} \bigr)_{i = 0, \dots, n}$ can be simulated exactly, for any choice of the sequence ${(u_i)}_{i = 0, \dots, n} \subset [T, T+\Delta]$, opens the way to Monte Carlo (MC) methods. It is natural to approximate the integral $\int_{T}^{T+\Delta}e^{X_{T}^{u}}\dd u$ using some discretization scheme $\sum_{i=0}^n \omega_i e^{X_T^{u_i}}$: this is the path we are going to follow. We are going to analyze the discretization error of a rectangle and a trapezoidal discretization scheme, which is our starting point to design optimal multilevel Monte Carlo (MLMC) estimators \citep{giles_2008} for $\VIX$ options.
In this respect, the present work can be seen as an extension to the setting of rough fractional processes of the work of \citep{lapeyre2001competitive}, who studied the discretization error of the rectangle and trapezoidal
schemes for the Asian payoff $A_T = \frac{1}{T}\int_{0}^{T} Y_{t}\dd t$ when $Y$ is a geometric Brownian motion,
%they showed that the weak and $L^{2}$ strong errors are of order $\cO(\frac{1}{n})$ for both an $n$-point rectangle and trapezoidal schemes 
and of the work of \citep{alaya2014multilevel}, who studied, for the same problem of Asian option pricing under geometric Brownian motion, the application  of the multilevel method.

The weak discretization error for the $\VIX$ in the rough Bergomi model has recently been studied by \citep{horvath2018volatility},
who showed that a $\VIX$ option \eqref{def:option:VIX2} with Lipschitz payoff $\varphi$ has a weak error of order $\cO\bigl(\frac 1 n \bigl)$ for a rectangle scheme on a uniform grid with $n$ points, resp.\ of order $\cO\bigl(\frac 1{n^2} \bigl)$ for a trapezoidal scheme on an adaptive grid.
We complement these results by stating upper bounds for the $L^p$ strong error (Proposition \ref{prop:strong:Lp:error}) and providing an explicit asymptotic expansion of the $L^2$ strong error of the rectangle scheme (Theorem \ref{thm:strong:L2:error}).
Let us note that, due to the cost $\cO(n^2)$ required for the exact simulation of a Gaussian vector with $n$ correlated components, the resulting strong and weak error asymptotics lead to a computational cost of order $\cO(\ve^{-4})$ (resp.\ $\cO(\ve^{-3})$) for a Monte Carlo estimator of $\VIX$ options based on a rectangle (resp.\ trapezoidal) scheme, when the overall $\MSE$ is required to be smaller than a given tolerance $\ve^2$.
We demonstrate, then, that the combination of a multilevel Monte Carlo estimator with the rectangle scheme allows to significantly reduce the 
computational cost to $\cO(\ve^{-2} \log^2(\ve))$, and further that the optimal asymptotical complexity $\cO(\ve^{-2})$
%of an unbiased estimator
can be obtained 
combining an MLMC estimator with the trapezoidal scheme, leading to significant computational time reduction (section \ref{sec:mlmc:vix:mlmc:estimator}).
We provide numerical tests supporting our theoretical findings.
%This paper is structured as follows. In section \ref{sec:mlmc:vix:rectangle:scheme}, we analyze the right-point rectangle and trapezoidal schemes for the $\VIX$ discretization in the rBergomi model, and provide some numerical tests. 
%In section \ref{sec:mlmc:vix:mlmc:estimator}, we investigate a multilevel estimator and demonstrate that it can reduce the complexity to $\cO(\ve^{-2})$.

It is worth mentioning that other approaches to the valuation of $\VIX$ derivatives under the class of log-normal models \eqref{def:sde:X_T}  have been recently explored:	\citep{Jacquier2018VIXFutures} investigate upper and lower bounds along with log-normal approximations for VIX futures, and \citep{Bonesini2021functional} develop a functional quantization for the log-forward variance process $u \mapsto X_T^u$.
Another stream of recent literature aims at computing asymptotic formulas in different asymptotic regimes: when the $\VIX$ option maturity $T$ is small and/or small volatility-of-volatility as in \citep{guyon2020vix} and \citep{Alos2018VIXskew}, or when the time-window $\Delta$ is small as in \citep[Chapter 6]{bourgey2020stochastic}, which is the case for the $\VIX$ index, for which $\Delta = 30 \, \mathrm{days} \approx \frac 1 {12}$.
Of course, the multilevel Monte Carlo estimator we design in section \ref{sec:mlmc:vix:mlmc:estimator} can be applied for any value of $T$ and $\Delta$, so that this valuation method is not limited to short-maturity options or to the VIX case, but allows to deal with options on general forward variances $V_T^{T_1,T_2}$ for any tenor structure $T \le T_1 \mino T_2$.

\section{Time discretization}

\label{sec:mlmc:vix:rectangle:scheme}

We focus on the rough Bergomi model for forward variances $\xi_t^u$, corresponding to the model \eqref{def:sde:X_T} equipped with the  power-law kernel $K\left(u,s\right)=  \frac \eta{\left(u-s\right)^{1/2 - H}}$.
Even though the rough Bergomi model is typically considered and analysed only 
for values of $H$ ranging in $\bigl(0,\frac 12 \bigr)$, for this interval contains the typical values of $H$ encountered in the calibration to the $\VIX$ and $\SP$ markets, there is of course no issue in defining the process \eqref{def:sde:X_T} for any value of $H \in (0,1)$.

An approximate simulation of the VIX random variable 
can be obtained via a discretization of the time integral in \eqref{eq:def:VIX2}. 
We fix a uniform grid with $n\in\bN^{*}$ time steps 
\begin{equation}\label{eq:def:uniform:grid}
	\mathcal{G}_{n}:=\left\{ u_{i}:=T + i\, h,\quad i=0,\dots,n\right\} ,\qquad h:=\frac{\Delta}{n},
\end{equation}
and consider the following simulation schemes. 
\begin{defn}[Discretization schemes for the $\VIX_{T}^{2}$] \label{def:discr}
The right-point rectangle scheme on the grid $\mathcal{G}_{n}$ is given by
\begin{equation}
\Rn:=\frac{1}{n}\sum_{i=1}^{n}e^{X_{T}^{u_{i}}},\label{eq:def:right:point:rectangle:scheme}
\end{equation}
and the trapezoidal scheme by
\begin{equation}
\Tn:=\frac{1}{2n}\sum_{i=1}^{n}\Bigl(e^{X_{T}^{u_{i}}}+e^{X_{T}^{u_{i-1}}}\Bigr).\label{eq:def:trapezoidal:scheme}	
\end{equation}
\end{defn}

The family of random variables $\left(X_{T}^{u_{i}}\right)_{i=0,\dots n}$ forms an $(n+1)$-dimensional Gaussian vector.
The mean $\mu=\left(\mu\left(u_{i}\right)\right)_{i=0,\dots,n}$ and covariance matrix $C=\left(C\left(u_{i},u_{j}\right)\right)_{0\leq i,j\leq n}$ are given by
%for all $0\leq i,j\leq n$, 
\begin{equation}
\mu\left(u_{i}\right):=\bE\left[X_{T}^{u_{i}}\right]=X_{0}^{u_{i}}-\frac{1}{2}\int_{0}^{T}K\left(u_{i},s\right)^{2}\dd s=X_{0}^{u_{i}}-\frac{\eta^{2}}{4H}\left(u_{i}^{2H}-\left(u_{i}-T\right)^{2H}\right),\label{eq:mean:ti}
\end{equation}
and 
\begin{equation}
C\left(u_{i},u_{j}\right):=\Cov\left(X_{T}^{u_{i}},X_{T}^{u_{j}}\right)=\int_{0}^{T}K\left(u_{i},s\right)K\left(u_{j},s\right)\dd s=\eta^{2}\int_{0}^{T}\left(u_{i}-s\right)^{H-\frac{1}{2}}\left(u_{j}-s\right)^{H-\frac{1}{2}}\dd s.\label{eq:variance:ti:tj}
\end{equation}
The diagonal terms of the covariance matrix can be computed explicitly
\[
C\left(u_{i},u_{i}\right)=\frac{\eta^{2}}{2H}\left(u_{i}^{2H}-\left(u_{i}-T\right)^{2H}\right),
\]
while the off-diagonal terms can be expressed in terms of the hypergeometric function $_{2}F_{1}(\cdot,\cdot;\cdot;\cdot)$ (see \citep[Chapter 15]{olver2010nist}): if $i < j$,
\begin{multline*}
C\left(u_{i},u_{j}\right) = \frac{\eta^{2}\left(u_{j}-u_{i}\right)^{H-\frac{1}{2}}}{H+\frac{1}{2}}\Bigl[
u_{i}^{H+\frac12}{}_{2}F_{1}\bigl(\frac{1}{2}-H,\frac{1}{2}+H;\frac{3}{2}+H;-\frac{u_{i}}{u_{j}-u_{i}}\bigr)\\
-\left(u_{i}-T\right)^{H+\frac12}{}_{2}F_{1}\bigl(\frac{1}{2}-H,\frac{1}{2}+H;\frac{3}{2}+H;-\frac{u_{i}-T}{u_{j}-u_{i}}\bigr)
\Bigr].
\end{multline*}
Standard numerical libraries provide an efficient evaluation of hypergeometric functions that can be used for an offline computation of the coefficients $C\left(u_{i},u_{j}\right)$ for $i < j$.
We used the function $\texttt{scipy.special.hyp2f1}$ from the $\texttt{scipy}$ library \citep{virtanen2020scipy} in our numerical tests.

In order to design optimal multilevel estimators in section \ref{s:MLMC}, we need to rely on asymptotic estimates of the strong and weak error associated to the discretization schemes in Definition \ref{def:discr}.

\subsection{Strong convergence}

The discretization error for the squared VIX \eqref{eq:def:VIX2} was studied by \citep{horvath2018volatility}, who showed in \citep[Proposition 2]{horvath2018volatility} that the bias
%(or weak error)
$\esp \bigl[\varphi(\VIX_{T}^{2})\bigr] - \esp\bigl[\varphi(\Rn) \bigr]$ for a VIX option with Lipschitz payoff $\varphi$  is of order $\cO\left(\frac{1}{n}\right)$.
Inspection of their proof reveals that their argument actually yields an 
estimate of the $L^1$ strong error $\esp \left[\bigl| \varphi(\VIX_{T}^{2}) - \varphi(\Rn) \bigr| \right]$.
Extending an $L^1$ estimate to an $L^p$ estimate for any positive $p$ is straightforward in this setting; therefore, following the steps in the proof of \citep[Proposition 2]{horvath2018volatility},  we obtain upper bounds on the $L^p$ strong errors for the two discretization schemes presented in Definition \ref{def:discr}, see Proposition \ref{prop:strong:Lp:error} below.

Note that, in general, the initial condition $(X_0^u)_{u \ge T}$ for the instantaneous forward variance curve is an arbitrary function; in practical applications, $X_0^u$ is bootstrapped from market options data, from the VIX futures term-structure or possibly from Variance Swap quotes (according to the use the stochastic model is intended for). In any of these cases, the regularity of the curve $u \mapsto X_0^u$ is a user-based choice -- typically, it is either a continuous or a piece-wise constant function. In any case, it is more than enough to assume that the function $X_0^\cdot	$ is locally bounded.

\begin{prop}[$L_{p}$ strong error; a slight extension of \citep{horvath2018volatility}, Proposition 2]
\label{prop:strong:Lp:error} 
Assume the initial forward variance curve $u \mapsto e^{X_{0}^u}$ is locally bounded.
Then, for any $T >0$ and $p>0$, 
\begin{align}
\left(\mathbb{E}\left[\left|\VIX_{T}^{2}-\Rn\right|^{p}\right]\right)^{\frac{1}{p}} & =\cO\left(\frac{1}{n}\right), \label{eq:estimate:lp:error:rectangle}\\
\left(\mathbb{E}\left[\left|\VIX_{T}^{2}-\Tn\right|^{p}\right]\right)^{\frac{1}{p}} & =\cO\left(\frac{1}{n^{1+H}}\right),
\label{eq:estimate:lp:error:trapezoidal}
\end{align}
as $n \to \infty$.
\end{prop}

\begin{rem}
Proposition \ref{prop:strong:Lp:error} holds for any value of $H$ in $(0,1)$.
\end{rem}

It is interesting to notice that the upper bound $\cO\left(\frac{1}{n}\right)$ on the strong rate of convergence of the rectangle scheme does not depend on the fractional index $H$; the dependence with respect to $H$ is only seen at the level of a more accurate quadrature method such as the trapezoidal scheme in \eqref{eq:estimate:lp:error:trapezoidal}.
\begin{rem} For comparison, the strong discretization error for the spot asset price $\dd S_t = S_t \, \sqrt{V_t} \, \dd Z_t$ is highly dependent on the value of $H$ in a model with instantaneous variance $V_t$ driven by a fractional process with H\"older regularity index $H$.
Here $Z$ denotes a standard Brownian motion.
\citep{neuenkirch2016order} show that an Euler--Maruyama scheme for $S$ based on $n$ discretization points has a strong error of order $\cO\Big(\frac{1}{n^H}\Big)$ when $V_t = e^{X_t}$ and $X$ is a fractional Ornstein--Uhlenbeck process.
Concerning weak errors, 
\citep{bayer2020weak} show that the weak error is $\cO\Big(\frac{1}{n^{H+\frac 12}}\Big)$  for payoffs $\varphi(S_T)$ with $\varphi \in \cC_b^p$ and $p= \lceil \frac 1H \rceil$ 
in a Gaussian stochastic volatility model where $V_t = X_t^t$ and $X_t^u$ is defined in \eqref{def:sde:X_T}.
\end{rem}	

Our next result is the exact asymptotics of the $L^{2}$ strong error for the rectangle scheme, which shows that the strong rate of convergence is precisely $\frac 1n$.
This result is additionally corroborated by some numerical experiments we present in the next section, which also provide numerical evidence that the strong rate of the trapezoidal scheme is precisely $\frac 1{n^{1+H}}$ (Figure \ref{fig:strong:weak:error}).

\begin{thm}[Exact asymptotics for the $L^{2}$ strong error, rectangle scheme]
\label{thm:strong:L2:error} Assume $H \in \bigl(0, \frac 12 \bigr)$ and that the initial instantaneous forward variance curve is kept constant over the VIX time window, that is $X_{0}^{u}=X_{0}$ for all $u\in\left[T,T+\Delta\right]$.
Then, as $n\to\infty,$ 
\begin{equation}\label{eq:L2:strong:error:rectangle}
\left(\mathbb{E}\left[\left|\VIX_{T}^{2}-\Rn\right|^{2}\right]\right)^{1/2}
\sim 
\frac{\Lambda\left(X_{0},T,\Delta,H\right)}{n}
\end{equation}
where 
\begin{equation*}
\Lambda\left(X_{0},T,\Delta,H\right)  :=
\frac{e^{X_{0}}}{2}\left(e^{\eta^{2}\frac{T^{2H}}{2H}}+e^{\eta^{2}\frac{\left(T+\Delta\right)^{2H}-\Delta^{2H}}{2H}}-2e^{\eta^{2}\int_{0}^{T}t^{H-\frac{1}{2}}\left(\Delta+t\right)^{H-\frac{1}{2}}\dd t}\right)^{1/2}.
\end{equation*}
\end{thm}

\noindent The proofs of Proposition \ref{prop:strong:Lp:error} and Theorem \ref{thm:strong:L2:error} are postponed to Appendix \ref{s:appendix}.
\medskip

\paragraph{Weak error.}
The discretization bias for the price of a VIX option $\esp[ \varphi\left(\VIX_{T}^{2}\right)]$ with Lipschitz payoff $\varphi$ can, of course, be upper bounded by the $L^1$ strong error multiplied by the Lipschitz constant of $\varphi$; it is a straightforward consequence of Proposition \ref{prop:strong:Lp:error}  (or Theorem \ref{thm:strong:L2:error} for the rectangle scheme)
that
\begin{align}
\left|\bE\left[\varphi\left(\VIX_{T}^{2}\right)-\varphi\left(\Rn\right)\right]\right| & =\cO\left(\frac{1}{n}\right),\label{eq:weak:error:rectangle}\\
\left|\bE\left[\varphi\left(\VIX_{T}^{2}\right)-\varphi\left(\Tn\right)\right]\right| & =\cO\left(\frac{1}{n^{1+H}}\right).\label{eq:weak:error:trapeze}
\end{align}
In our numerical experiments in section \ref{s:numerics} we can precisely observe the asymptotic behavior predicted by \eqref{eq:weak:error:rectangle}-\eqref{eq:weak:error:trapeze},
in the case of VIX call options.
\begin{rem}\label{rem:lipschitz:payoff}
A $\VIX$ call option with strike $\kappa>0$ corresponds to the payoff $\varphi(x) = (\sqrt{x} - \kappa)_+$. In such a case, the function $\varphi$ is Lipschitz, for it coincides with the integral of its bounded derivative $\varphi'$ given by $\varphi'(x)  = \frac {1}{2 \sqrt x} \leq \frac {1}{2\kappa}$ for $x \maj \kappa^2$ and $\varphi'(x) =0$ for $x \mino \kappa^2$.
\end{rem}

\begin{rem}[Non-uniform grids] \label{rem:non_unif_grid}
Using a non-uniform discretization grid such as 
\be  \label{e:non_uniform_grid}
\biggl\{ u_{i}:=T+ \Delta \biggl(\frac i n\biggr)^{a}, \quad i=0,\dots,n\biggr\} , \quad a > 0,
\ee
one can improve the asymptotic behavior of the weak error: according to \citep[Corollary 1]{horvath2018volatility}, the weak error associated to the trapezoidal scheme based on the grid \eqref{e:non_uniform_grid} is $\cO(n^{-2})$ instead of $\cO(n^{-(1+H)})$ in \eqref{eq:weak:error:trapeze}, provided that $a >\frac{2}{H+1}$. 
Note that with such a value of $a$, the resulting mesh \eqref{e:non_uniform_grid} gets finer and finer as $u_i$ approaches $T$, in order to compensate for the explosion of the integration kernel $K\left(u, s\right)= \frac \eta {\left(u-s\right)^{\frac{1}{2}-H}}$ when $u$ approaches $T$ from above in  \eqref{def:option:VIX2} and $s$ approaches $T$ from below in \eqref{def:sde:X_T}.
It will be seen in section \ref{s:MLMC} that the optimal complexity $\cO\left(\ve^{-2}\right)$ for an estimator with mean-squared error $\ve^2$ can already be attained combining the multilevel method with a discretization scheme based on the uniform grid $\cG_n$ in \eqref{eq:def:uniform:grid}.
Appealing to the more complex mesh choice \eqref{e:non_uniform_grid} is therefore not necessary in order to construct estimators of $\VIX$ options with optimal rate of convergence.
\end{rem}

\subsection{Control variate\label{subsec:control:variate} for VIX option pricing}
In order to obtain a reliable Monte-Carlo estimate of the weak error, we need to employ a variance reduction method. 
Due to the structure of the problem, there is an efficient control variate technique (in the spirit of the control variate for Asian option pricing in \citep{kemna1990pricing}), already exploited by \citep{horvath2018volatility}, that we now present in detail.

Noticing that the discretization schemes $\Rn$ and $\Tn$ in \eqref{eq:def:right:point:rectangle:scheme} and \eqref{eq:def:trapezoidal:scheme} are given by means of the exponentials $e^{X_{T}^{u_i}}$, it is reasonable to consider as a control variate the exponential of the mean $e^{\frac{1}{n}\sum_{i=1}^{n}X_{T}^{u_i}}$ -- this choice is appropriate notably because the parameter $u_i $ is restricted to a small time window $\Delta \approx \frac 1{12}$.
The key ingredient is that the control variate $e^{\frac{1}{n}\sum_{i=1}^{n}X_{T}^{u_i}}$ is lognormal: the Gaussian random variable $\frac{1}{n}\sum_{i=1}^{n}X_{T}^{u_i}\egl\cN\bigl(\mu^{n},\left(\sigma^{n}\right)^{2}\bigr)$ has explicit characteristics
\begin{equation}
\begin{split}\mu^{n} & =\frac{1}{n}\sum_{i=1}^{n}\mu\left(u_{i}\right)=\frac{1}{n}\sum_{i=1}^{n}\left(X_{0}^{u_{i}}-\frac{1}{2}\int_{0}^{T}K\left(u_{i},s\right)^{2}\dd s\right),\\
\sigma^{n} & =\sqrt{\frac{1}{n^2}\sum_{i,j=1}^{n}C\left(u_{i},u_{j}\right)},
\end{split}
\label{eq:mupn:sigmapn}
\end{equation}
so that the prices of call or put options on the control variate boil down to a Black-Scholes formula.	
When considering the right-point rectangle scheme \eqref{eq:def:right:point:rectangle:scheme},
we estimate option prices on the $\VIX^{2}$ with 
\be \label{e:option_price_with_cv}
\bE\biggl[\varphi\biggl(\frac{1}{n}\sum_{i=1}^{n}e^{X_{T}^{u_{i}}}\biggr)\biggr]
\approx
\frac{1}{M}\sum_{m=1}^{M}
\biggl[\varphi\biggl(\frac{1}{n}\sum_{i=1}^{n}e^{X_{T}^{u_{i},m}}\biggr)-\varphi\Big(e^{\frac{1}{n}\sum_{i=1}^{n}X_{T}^{u_{i},m}}\Big)
\biggr]
+ \mathrm{CV}_{n},
\ee
where $\left(X_{T}^{u_{1},m},\dots,X_{T}^{u_{n},m}\right)_{1\leq m\leq M}$ are  $M$ independent Monte Carlo samples of the Gaussian vector $\left(X_{T}^{u_{1}},\dots,X_{T}^{u_{n}}\right)$ and $\mathrm{CV}_{n}:=\bE\left[\varphi\left(e^{\frac{1}{n}\sum_{i=1}^{n}X_{T}^{u_{i}}}\right)\right]$.
In the case of VIX call and put options, $\mathrm{CV}_{n}$ is given by
\be \label{e:explicit_price_CV}
\begin{aligned}
\mathrm{CV}_{n}=\begin{cases}
C_{\BS}\left(e^{\frac{\mu^{n}}{2}+\frac{\left(\sigma^{n}\right)^{2}}{8}},\kappa,\frac{\sigma^{n}}{2}\right)
& \text{if }\varphi\left(x\right)=\left(\sqrt{x}-\kappa\right)_{+},
\\
P_{\BS}\left(e^{\frac{\mu^{n}}{2}+\frac{\left(\sigma^{n}\right)^{2}}{8}},\kappa,\frac{\sigma^{n}}{2}\right) & \text{if }\varphi\left(x\right)=\left(\kappa-\sqrt{x}\right)_{+},\\
e^{\frac{\mu^{n}}{2}+\frac{\left(\sigma^{n}\right)^{2}}{8}} & \text{if }\varphi\left(x\right)=\sqrt{x},
\end{cases}
\end{aligned}
\ee
where $C_{\BS}$ and $P_{\BS}$ are the standard Black--Scholes call and put prices $C_{\BS}\left(x,y,z\right) =x\,\Phi\left(\frac{\ln\left(x/y\right)}{z}+\frac{z}{2}\right) - y \,\Phi\left(\frac{\ln\left(x/y\right)}{z}-\frac{z}{2}\right)$ and $P_{\BS}\left(x,y,z\right) =y \, \Phi\left(-\frac{\ln\left(x/y\right)}{z}+\frac{z}{2}\right) - x\,\Phi\left(-\frac{\ln\left(x/y\right)}{z}-\frac{z}{2}\right)$, $\Phi$ being the c.d.f.\ of the standard normal distribution.
For the trapezoidal scheme \eqref{eq:def:trapezoidal:scheme}, we use 
\be \label{e:CV_trap}
\frac{1}{2}\Big[
\varphi\left(e^{\frac{1}{n}\sum_{i=1}^{n}X_{T}^{u_{i}}}\right)
+ \varphi\left(e^{\frac{1}{n}\sum_{i=1}^{n}X_{T}^{u_{i-1}}}\right)\Big]
\ee
as a control variate, where both the variables $e^{\frac{1}{n}\sum_{i=1}^{n}X_{T}^{u_{i}}}$ and $e^{\frac{1}{n}\sum_{i=1}^{n}X_{T}^{u_{i-1}}}$ are log-normal.
For call and put options, the expectation $\esp \bigl[ e^{\frac{1}{n}\sum_{i=1}^{n}X_{T}^{u_{i-1}}} \bigr]$ has a similar expression to  \eqref{e:explicit_price_CV}.

\subsection{Numerical tests} \label{s:numerics}

From now on, we will use the notation $\Dn$ to denote either the right-point rectangle scheme $\Rn$ when $\cD_{n}=\cR_{n}$ or the trapezoidal scheme $\Tn$ when $\cD_{n}=\cT_{n}$.
In order to check the $L^{2}$ strong
error asymptotics \eqref{eq:L2:strong:error:rectangle}
for the rectangle scheme and the $L^{2}$ estimate \eqref{eq:estimate:lp:error:trapezoidal}
for the trapezoidal scheme, we consider three different sets of parameters $X_{0},\eta,T,H$ (specified in the legends of Figures \ref{fig:strong:weak:error}  and \ref{fig:strong:trapeze}) and set $\Delta=\frac{1}{12}$.
Since we cannot simulate exactly the random variable $\VIX_{T}^{2}$,
we approximate $\VIX_{T}^{2}$ with $\VIX_{T}^{2,\cD_{n_{\rm ref}}}$ using a high number ${n_{\rm ref}}$ of discretization points, and approximate the strong error or the scheme $\Dn$ with a Monte Carlo estimate of $\bE\left[\left|\VIX_{T}^{2,\cD_{n_{{\rm ref}}}}-\Dn\right|^{2}\right]^{\frac{1}{2}} \approx \bE\left[\left|\VIX_{T}^{2}-\Dn\right|^{2}\right]^{\frac{1}{2}}$ for values of $n \ll n_{{\rm ref}}$.
We choose $n_{{\rm ref}}=2000$ and generate $M=10^{5}$ Monte Carlo samples. 
\begin{figure}[t]
	\begin{centering}
		\includegraphics[scale=0.5]{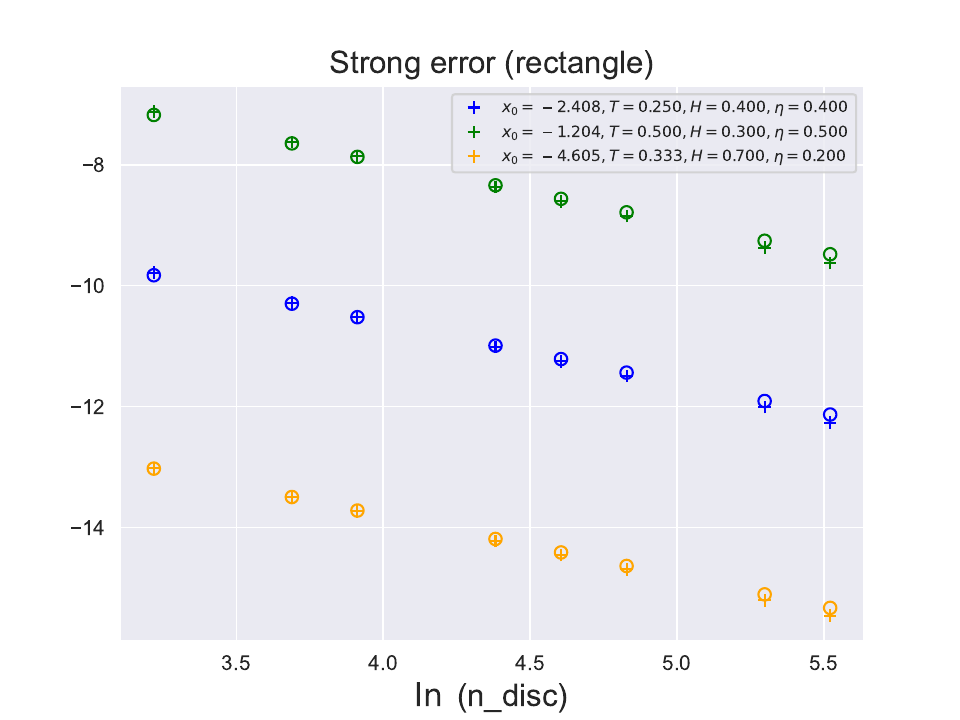}
		\hspace{-5mm}
		\includegraphics[scale=0.5]{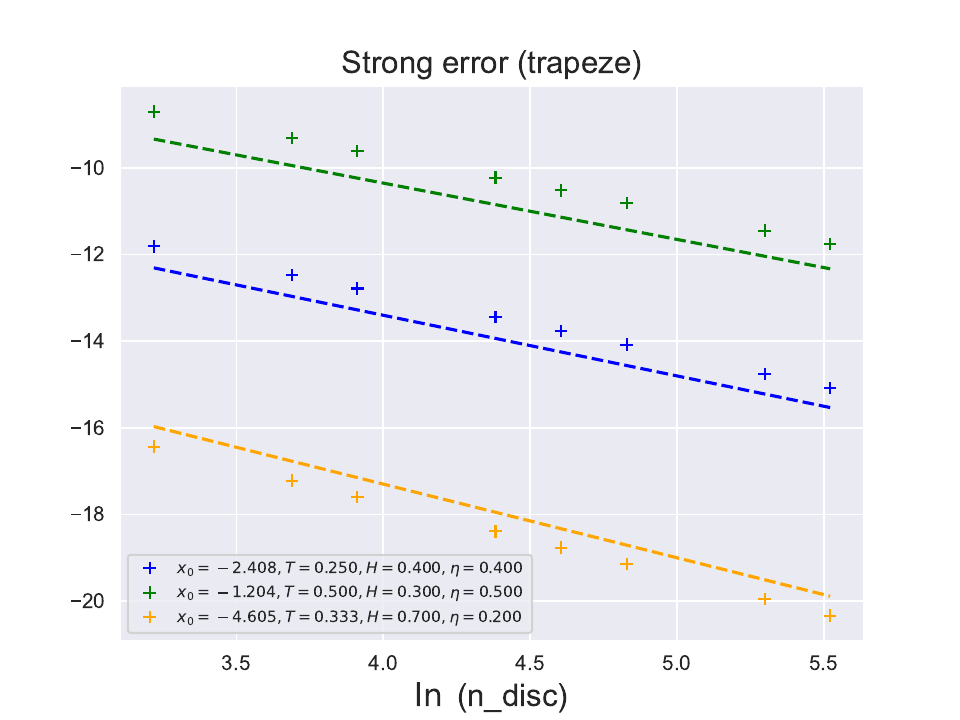} 
		\par\end{centering}
	\caption{\emph{Left}: log-log plot of a Monte Carlo estimate of the $L^{2}$ strong error of the rectangle scheme $\Rn$ (marker ``$+$'') together with its asymptotic expansion in \eqref{eq:L2:strong:error:rectangle} (marker ``$\circ$'').
		\emph{Right}: Monte Carlo estimate of the $L^{2}$ strong error
		of the trapezoidal scheme $\Tn$ (marker ``$+$'') and, for comparison, straight lines with slope $-(1+H)$.}
	\label{fig:strong:weak:error} 
\end{figure}

\noindent Figure \ref{fig:strong:weak:error} confirms the asymptotic expansion obtained in Theorem
\ref{thm:strong:L2:error} for the rectangle scheme and the $\cO(n^{-(1+H)})$ estimate \eqref{eq:estimate:lp:error:trapezoidal} for the trapezoidal scheme.
Note also that, as expected,
the error of the trapezoidal scheme
is smaller than that of the rectangle scheme.

We also want to check the weak convergence rate  for VIX call options.
We estimate the weak error for a call option $(\VIX_{T} - \kappa)^+$ with parameters $X_{0}=-2.896 \approx \ln\left(0.235^{2}\right)$, $\eta=0.5$, $H=0.3$, $T=\frac{1}{4}$, strike price $\kappa=0.1$, and consider discretisation grids with
$n\in\left[5, 14\right]$ points.
The control variate described in section \ref{subsec:control:variate} is used to reduce the variance of the Monte Carlo estimation. 
The reference price $\esp\left[ (\VIX_{T} - \kappa)^+ \right] \approx 0.13093742\pm5\times10^{-8}$ is estimated via
a rectangle scheme with $400$ discretization points and $3\times10^{6}$ Monte Carlo samples, again using the control variate in section \ref{subsec:control:variate}.

\noindent 
\begin{figure}[t]
	\begin{centering}
		\includegraphics[scale=0.5]{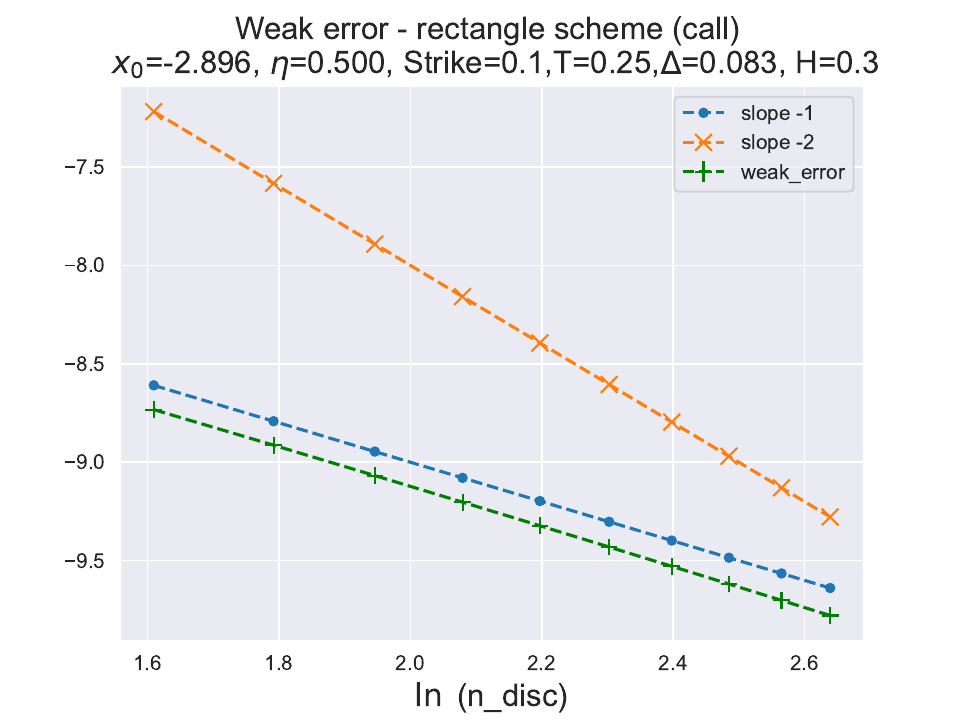}
		\hspace{-5mm}
		\includegraphics[scale=0.5]{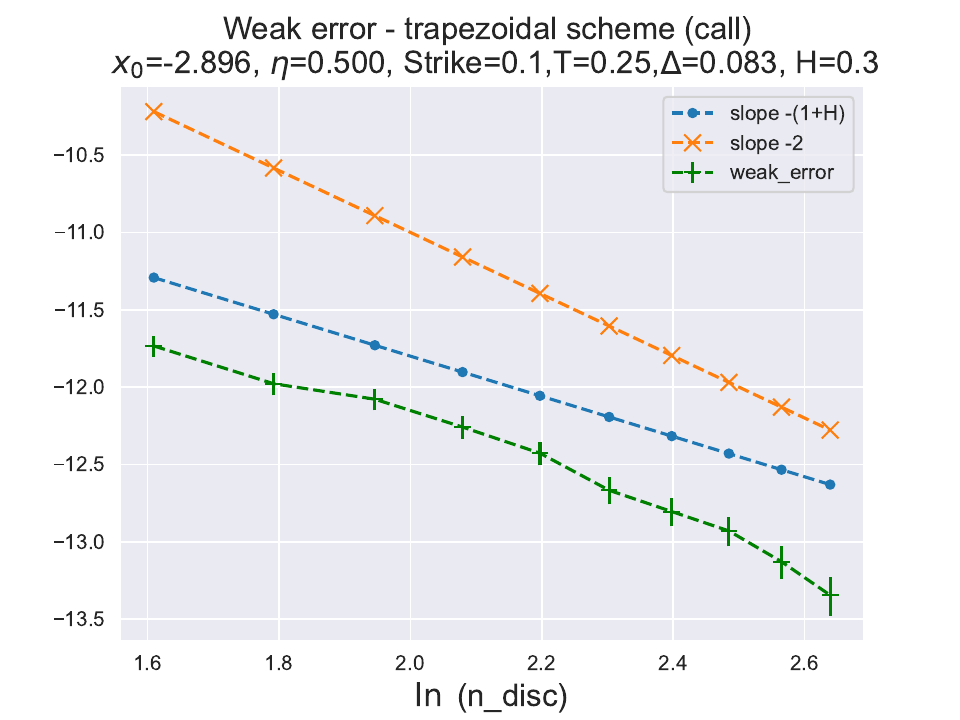} 
		\par\end{centering}
	\caption{\emph{Left}: log-log plot of a Monte-Carlo estimate of the weak error for a call option on the rectangle scheme $\Rn$ (marker ``$+$'') and, for comparison, straight lines with slopes equal to $-1$ and $-2$ (marker ``$--$'').
		\emph{Right}:  weak error for a call option on the trapezoidal scheme $\Tn$
		%marker ``$+$''
		and, for comparison, straight lines with slopes $-(1+H)$ and $-2$ (using the same marker convention as in the left pane).}
	\label{fig:strong:trapeze} 
\end{figure}

\noindent In Figure \ref{fig:strong:trapeze}, we can observe the weak convergence rates $\frac 1 n$ and $\frac 1 {n^{1+H}}$ for the two discretization schemes, in alignment with estimates \eqref{eq:weak:error:rectangle} and \eqref{eq:weak:error:trapeze} (the VIX call payoff being a Lipschitz function of the VIX squared, as discussed in Remark \ref{rem:lipschitz:payoff}).

\section{Multilevel Monte Carlo}

\label{sec:mlmc:vix:mlmc:estimator}

\subsection{Standard Monte-Carlo estimation}

Once  the simulation scheme $\Dn$ corresponding to either the rectangle $\Rn$ or the trapezoidal scheme $\Tn$ has been fixed, the plain Monte Carlo estimator of the expectation $\bE\bigl[\varphi\bigl(\VIX_T^2 \bigr)\bigr]$,
\be  \label{e:MC}
\widehat{P}_{M}^{\cD_{n}}=\frac{1}{M}\sum_{m=1}^{M}\varphi\Bigl(\VIX_{T,m}^{2,\cD_{n}}\Bigr),
\ee 
based on $M$ independent copies $\bigl( \VIX_{T,m}^{2,\cD_{n}} \bigr)_{1 \le m \le M}$ of $\Dn$, has a 
mean-squared error that can be decomposed in a bias and a variance term
\begin{equation*}
	\MSE 
	=
	\bE\left[\left(\bE\left[\varphi\left(\VIX_{T}^{2}\right)\right]-\widehat{P}_{M}^{\cD_{n}}\right)^{2}\right]\\
	=
	\bE\left[\varphi\left(\VIX_{T}^{2}\right)-\varphi\left(\Dn\right)\right]^{2}
	+ \frac 1 M \Var\left(\varphi\left(\VIX_{T}^{2,\cD_{n}}\right)\right).
\end{equation*}
The construction of $\Dn$ requires to sample the $(n+1)$-dimensional Gaussian vector $\left(X_{T}^{u_{i}}\right)_{i=0,\dots n}$ with given
characteristics \eqref{eq:mean:ti} and \eqref{eq:variance:ti:tj}.
We consider exact simulation of $\left(X_{T}^{u_{i}}\right)_{i}$ based on the Cholesky decomposition $L_n L^T_n$
of the covariance matrix \eqref{eq:variance:ti:tj}.
We assume that the required Cholesky matrix $L_n$ for each desired value of $n$ has been computed once and for all as offline work, so that sampling a $(n+1)$-dimensional Gaussian vector requires $\cO\left(n^{2}\right)$ operations due to the matrix multiplication $L_n G$, with $G \sim \mathcal N(0, \mathrm{Id}_n)$.\footnote{Unfortunately, circulant embedding methods (see e.g.\ \citep[Chapter XI, section 3]{asmussen2007stochastic}), which rely on the fast Fourier transform and achieve a complexity of $\cO\left(n\log(n)\right)$ in the simulation of stationary Gaussian processes, do not seem to be straightforward to apply here since the process $(X_{T}^u)_{u \ge T}$ and its increment process fail to be stationary.}
The construction of $\Dn$ in \eqref{eq:def:trapezoidal:scheme} or \eqref{eq:def:right:point:rectangle:scheme}
from $\left(X_{T}^{u_{i}}\right)_{i}$ requires an additional sum over $n$ points, at a cost $\cO\left(n\right)$.
If we want the mean-squared error to satisfy $\MSE\leq\ve^{2}$
for a given accuracy $\ve>0$, we have to set $M=\cO\left(\ve^{-2}\right)$.
If $\varphi$ is Lipschitz, then according to the weak error behavior in \eqref{eq:weak:error:rectangle} and \eqref{eq:weak:error:trapeze}, we choose $n=\cO\left(\ve^{-1}\right)$ when $\cD_{n}=\cR_{n}$ (rectangle scheme) and $n=\cO\left(\ve^{-\frac{1}{1+H}}\right)$ when $\cD_{n}=\cT_{n}$ (trapezoidal scheme). 
Consequently, the total computational cost for the plain Monte Carlo estimator \eqref{e:MC} is, when $\cD_{n}=\cR_{n}$,
\[
\Cost_{\mathrm{MC}}^{\cR_{n}}=
\left(\cO\left(n^{2}\right)+\cO\left(n\right)\right)\times M
=\cO\left(\ve^{-4}\right),
\]
while, when $\cD_{n}=\cT_{n}$, 
\be \label{e:cost_MC_trap}
\Cost_{\mathrm{MC}}^{\cT_{n}}
=
\left(\cO\left(n^{2}\right)+\cO\left(n\right)\right)\times M
=
\cO\left(\ve^{-2 \left(1+\frac{1}{1+H} \right)}\right).
\ee
Note that the exponent $2 \left(1+\frac{1}{1+H} \right)$ takes values between $-4$ (when $H \to 0$) and $-3$ (when $H \to 1$). 
With reference to Remark \ref{rem:non_unif_grid}, considering a Monte-Carlo estimator \eqref{e:MC} based on a trapezoidal scheme with the non-uniform grid \eqref{e:non_uniform_grid} and $a (1+H) \maj 2$, the same error analysis leads to an overall computational cost of order $\cO(\ve^{-3})$ under the constraint $\MSE\leq\ve^{2}$, which is an improvement with respect to \eqref{e:cost_MC_trap}, but still far from the asymptotically optimal complexity $\cO(\ve^{-2})$.

\subsection{Multilevel scheme} \label{s:MLMC}

Let $L\in\mathbb{N}^{*}$ and let $\bm{n}=\left(n_{0},\dots n_{L}\right)\in\left(\bN^{*}\right)^{L+1}$ and $\bm{M}=\left(M_{0},\dots,M_{L}\right)$ be multi-indexes 
representing respectively an increasing sequence of time steps
and a sequence of Monte Carlo sample sizes.
For the ease of notation,
we set for all $\ell=0,\dots,L$, 
\[
P_{\ell}^{\cD}:=\varphi\left(\VIX_{T}^{2,\cD_{n_{l}}}\right).
\]
The index $\ell$ refers to the discretization \emph{level}; $P_{\ell}^{\cD}$ is an approximation of the VIX option payoff $\varphi(\VIX_T^2)$ at level $\ell$.
Given $n_{0}\in\bN^{*}$, we set 
\[
n_{\ell} = n_{0}2^{\ell}, \qquad \ell=0,\dots,L ,
\]
which means we double the number of time steps from one level to the next.
This implies that, once we have generated the $(n_\ell+1)$-dimensional Gaussian sample $(X_T^{u_i})_{u_i \in \mathcal G_{n_\ell}}$ entering into $P_{\ell}^{\cD}$ for some $\ell \ge 1$,  we can extract a $(n_{\ell-1}+1)$-dimensional Gaussian sample that we use to construct $P_{\ell-1}^{\cD}$, by selecting the $X_T^{u_{2 k}}$ for $k = 1, \dots, n_0 2^{\ell-1}$.

The multilevel method 
is inspired by the decomposition of the bias of $P_{L}^{\cD}$ as a telescopic sum:
\[
\bE\left[P_{L}^{\cD}\right]=\bE\left[P_{0}^{\cD}\right] 
+ \sum_{\ell=1}^{L} \bE\left[P_{\ell}^{\cD}-P_{\ell-1}^{\cD}\right].
\]
We can use independent Monte Carlo samples to approximate each term $\bE\left[P_{\ell}^{\cD}-P_{\ell-1}^{\cD}\right]$.
The resulting multilevel estimator is given by
\begin{equation}
	\PML^{\cD}:=\frac{1}{M_{0}}\sum_{m=1}^{M_{0}}P_{0}^{\cD,(0,m)}+\sum_{\ell=1}^{L}\frac{1}{M_{\ell}}\sum_{m=1}^{M_{\ell}}\left(P_{\ell}^{\cD,(\ell,m)}-P_{\ell-1}^{\cD,(\ell,m)}\right),\label{eq:def:mlmc:estimator}
\end{equation}
where for every $\ell$, the random variables $(P_{\ell}^{\cD,(\ell,m)})_{1\leq m\leq M_{\ell}}$
and $(P_{\ell-1}^{\cD,(\ell,m)})_{1\leq m\leq M_{\ell}}$ are independent
copies of $P_{\ell}^{\cD}$ and $P_{\ell-1}^{\cD}$ respectively.
For a fixed level $\ell$, the random variables $P_{\ell}^{\cD,(\ell,m)}$ and $P_{\ell-1}^{\cD,(\ell,m)}$ are constructed using the same Gaussian samples, according to the procedure described above.
%we extract $P_{\ell-1}^{\cD,(\ell,m)}$ from a given sample $P_{\ell}^{\cD,(\ell,m)}$.
Note that $\bE\bigl[\PML^{\cD}\bigr]=\bE\left[P_{L}^{\cD}\right]$ and, by independence,
\[
\Var\Bigl(\PML^{\cD}\Bigr)=\sum_{\ell=0}^{L}\frac{V_{\ell}^{\cD}}{M_{\ell}},
\]
where 
\[
%V_{0}^{\cD}:=\Var\left(P_{0}^{\cD}\right), \quad 
V_{\ell}^{\cD}:=\Var\left(P_{\ell}^{\cD} - P_{\ell-1}^{\cD} \right), \qquad \forall \, \ell=1,\dots,L \,,
\]
and we set $P_{-1}^{\cD} = 0$ so that $V_{0}^{\cD}=\Var\left(P_{0}^{\cD}\right)$.
We denote the cost of generating one sample of $P_{\ell}^{\cD}-P_{\ell-1}^{\cD}$
by 
\[
%C_{0}^{\cD}:=\Cost\left( P_{0}^{\cD} \right), \quad 
C_{\ell}^{\cD}:=\Cost\left(P_{\ell}^{\cD} - P_{\ell-1}^{\cD}\right), \qquad \forall \, \ell=1,\dots,L \,.
\]
As pointed out in the previous section, the cost of generating $P_{\ell}^{\cD}$ is $\cO\left(n_{\ell}^{2}\right)$ for both discretization schemes $\cD=\cR$ and $\cD=\cT$, due to the Cholesky exact simulation method we consider to sample the Gaussian vector $(X_T^{u})_{u \in \mathcal G_{n_\ell}}$, hence $C_{\ell}^{\cD} = \cO\left(n_{\ell}^{2}\right)$.
The overall cost of the multilevel estimator in \eqref{eq:def:mlmc:estimator} is $C^\cD_{\bm{M},\bm{n}} := \sum_{\ell = 1}^L M_l \, C_{\ell}^{\cD}$.

Denote $\MSE^\cD_{\bm{M},\bm{n}} := \esp \Bigl[ \Bigl(\PML^\cD - \esp[\varphi(\VIX_T)]\Bigl)^2 \Bigr]$ the  mean-squared error of the multilevel estimator; we want this MSE to be smaller than a given threshold $\ve^2$. 
The architecture of the optimal multilevel estimator  -- that is, the number $L$ of levels and the allocation of the simulation effort $M_\ell$ at each level -- can be identified by solving the constrained optimization problem
\be \label{e:optim_ML}
\min_{M_0, \dots, M_L}
C^\cD_{\bm{M},\bm{n}}
\quad \mbox{ such that } \quad
\MSE^\cD_{\bm{M},\bm{n}} \le \ve^2 \,.
\ee 

We now prove that \eqref{e:optim_ML} leads to a multilevel estimator achieving $\MSE^\cD_{\bm{M},\bm{n}} \leq\ve^{2}$ with computational complexity $\cO\big(\ln(\ve)^{2}\ve^{-2}\big)$ when $\cD$ is the right-point rectangle scheme and with the optimal complexity rate $\cO\big(\ve^{-2}\big)$ when $\cD$ is the trapezoidal scheme. 

\begin{thm}
	\label{thm:mlmc:discetization:integral}Suppose that the payoff function
	$\varphi$ is Lipschitz, and consider $n_{0}\in\bN^{*}$. 
	Then, for every tolerance $\ve>0$ there exist an initial number of samples $M_{0,\cR}$ and a number of levels $L_{\cR}$
	(resp.\ $M_{0,\cT}$ and $L_{\cT}$)
	%$\in\bN^{*}$
	such that the rectangle (resp.\ trapezoidal) multilevel estimator $\PML^{\cR}$ (resp.\ $\PML^{\cT}$),
	defined in \eqref{eq:def:mlmc:estimator}, has a mean-squared error satisfying $\MSE\leq\ve^{2}$ and a computational complexity $\cO\big(\ln\left(\ve\right)^{2}\ve^{-2}\big)$
	(resp.\ $\cO\left(\ve^{-2}\right)$), obtained setting
	\be \label{e:opt_ML_rect}
	n_{\ell}=n_{0}2^{\ell},
	\qquad M_{\ell,\cR} = M_{0,\cR} \, 2^{-2\ell},
	\qquad \forall \, \ell = 0, \dots, L_{\cR},
	\ee
	for the rectangle scheme, and respectively 
	\be \label{e:opt_ML_trap}
	n_{\ell}=n_{0}2^{\ell},
	\qquad M_{\ell,\cT} = M_{0,\cT} \, 2^{-(2+H)\ell},
	\qquad \forall \, \ell = 0, \dots,  L_{\cT},
	\ee
	for the trapezoidal scheme. 
\end{thm}

\begin{rem}[The parameters $M_{0,\cR}$ and $L_{\cR}$]
	\label{rem:multilevel:free:parameters}
	Using \eqref{e:opt_ML_rect} and \eqref{e:opt_ML_trap}, the multilevel estimators in Theorem \ref{thm:mlmc:discetization:integral} are fully specified once the number of levels $L$ and the initial number of samples $M_0$ have been fixed.
	We can choose $L_\cR = \cO(\ln(1/\ve))$ and $M_{0,\cR} = \cO(\ve^{-2}\ln(1/\ve))$ for the rectangle scheme, and $L_\cT = \cO(1)$ and $M_{0,\cT} = \cO(\ve^{-2})$ for the trapezoidal scheme.
	In the former case, we can precisely set 
	$L_{\cR}=\left\lceil \frac{\ln\left(\sqrt{2} \, c_{1}\ve^{-1}\right)}{\ln\left(2\right)}\right\rceil$ and
	$M_{0,\cR}=\left\lceil \frac{2}{\ve^{2}} \, c_{2}\left(L_{\cR}+1\right)\right\rceil$,
	where $c_{1}=\Lip \, \Lambda \, n_{0}^{-1}$,
	$c_{2}=10 \, \Lip^{2}\Lambda^{2} \, n_{0}^{-2}$, $\Lip$ is the Lipschitz constant of the payoff function, and the constant $\Lambda = \Lambda\left(X_{0},T,\Delta,H\right)$ is defined in Theorem \ref{thm:strong:L2:error}.
	Using the exact asymptotic of the $L^{2}$ strong error in Theorem \ref{thm:strong:L2:error}, we can see that the associated multilevel estimator $\PML ^{\cR}$ satisfies the requirement $\MSE\leq\ve^{2}$ with a computational complexity of order $\cO\bigl(\ln\left(\ve\right)^{2}\ve^{-2}\bigr)$.
	We will use the specifications above of the parameters $L_\cR$ and $M_{0, \cR}$ in our numerical tests in section \ref{s:numerics_ML}.
\end{rem}

\begin{proof}[Proof of Theorem \ref{thm:mlmc:discetization:integral}]
	We start with the rectangle-based multilevel estimator $\PML^{\cR}$.
	We borrow the mathematical notations from \citep[Theorem 1]{giles_2015}.
	As discussed above,
	the computational cost $C_{l}^{\cR}$ for each level $\ell$ is of order $\cO\left(n_{\ell}^{2}\right)$ and therefore, with reference to the notation of  \citep[Theorem 1]{giles_2015}, we have $\gamma=2$.
	We have to estimate the variance $V_{\ell}^{\cR}$ at level $\ell$.   
	Since $\varphi$ is Lipschitz (with Lipschitz constant $\Lip$), from the strong error estimate \eqref{eq:L2:strong:error:rectangle} it follows that
	\begin{align*}
		V_{\ell}^{\cR} \leq\bE\left[\left(P_{\ell}-P_{\ell-1}\right)^{2}\right]
		& =\bE\left[\left(\varphi\left(\VIX_{T}^{2,\cR_{\ell}}\right)-\varphi\left(\VIX_{T}^{2}\right)+\varphi\left(\VIX_{T}^{2}\right)-\varphi\left(\VIX_{T}^{2,\cR_{\ell-1}}\right)\right)^{2}\right]
		\\
		& \leq2 \, \Lip^{2} \, \bE\left[\left(\VIX_{T}^{2}-\VIX_{T}^{2,\cR_{\ell}}\right)^{2}+\left(\VIX_{T}^{2}-\VIX_{T}^{2,\cR_{\ell-1}}\right)^{2}\right]
		\\
		& \leq2 \, \Lip^{2} \, \Lambda\left(X_{0},T,\Delta,H\right)^{2} (1 + o(1))
		\left(n_{\ell}^{-2}+n_{\ell-1}^{-2}\right)
		\le c_{2} \, 2^{-\beta\ell}
	\end{align*}
	for $n_\ell$ large enough, where we have set $c_{2} = 10\, \Lip^{2} \, \Lambda\left(X_{0},T,\Delta,H\right)^{2}n_{0}^{-2}$ and $\beta=2$.
	We then know how to set the optimal number of samples $M_{\ell,\cR}$
	at level $\ell$: we set $M_{\ell,\cR}=M_{0,\cR} \, 2^{-\frac{\left(\beta+\gamma\right)}{2}\ell} = \, M_{0,\cR}2^{-2\ell}$.
	The value of $M_{0,\cR}$ has to fixed in such a way that $\Var\bigl(\PML^{\cR}\bigr)\leq\frac{\ve^{2}}{2}$, that is
	\begin{align*}
		\Var\bigl(\PML^{\cR}\bigr)
		& =\sum_{\ell=0}^{L_{\cR}}\frac{V_{\ell}^{\cR}}{M_{\ell,\cR}}\leq\frac{c_{2}}{M_{0,\cR}}\sum_{\ell=0}^{L_{\cR}}2^{-2\ell+2\ell}=\frac{c_{2}\left(L_{\cR}+1\right)}{M_{0,\cR}},
	\end{align*}
	which leads to $M_{0,\cR}=\left\lceil \frac{2}{\ve^{2}}c_{2}\left(L_{\cR}+1\right)\right\rceil .$
	Now,  it follows from Theorem \ref{thm:strong:L2:error} and the Cauchy--Schwarz inequality that
	\[
	\left|\bE\left[\varphi\left(\VIX_{T}^{2,\cR_{\ell}}\right)-\varphi\left(\VIX_{T}^{2}\right)\right]\right|\leq\bE\left[\left|\varphi\left(\VIX_{T}^{2,\cR_{\ell}}\right)-\varphi\left(\VIX_{T}^{2}\right)\right|^{2}\right]^{\frac{1}{2}}
	\le c_{1} 2^{-\alpha\ell}
	= c_{1} 2^{-\alpha\ell},
	\]
	where we have set $c_{1}=\Lip \, \Lambda\left(X_{0},T,\Delta,H\right)n_{0}^{-1}$
	and $\alpha=1$. 
	The bias of the multilevel estimator therefore satisfies 
	\[
	\left|\esp\left[\PML^{\cR}  - \varphi\left(\VIX_{T}^{2}\right) \right] \right|
	= 
	\left|\bE\left[\varphi\left(\VIX_{T}^{2,\cR_{L_{\cR}}}\right) - \varphi\left(\VIX_{T}^{2}\right)\right]\right|\leq c_{1}2^{-L_{\cR}} \,.
	\]
	Imposing the squared bias to be smaller than $\frac{\ve^{2}}{2}$, that is
	\[
	c_{1}^{2} \, 2^{-2L_{\cR}}\leq\frac{\ve^{2}}{2}
	\iff
	L_{\cR}\geq\frac{\ln\left(\sqrt{2} \, c_{1}\, \ve^{-1}\right)}{\ln\left(2\right)},
	\]
	we set $L_{\cR}=\left\lceil \frac{\ln\left(\sqrt{2} \, c_{1} \, \ve^{-1}\right)}{\ln\left(2\right)}\right\rceil$.
	Overall, we are in the case $\alpha\geq\frac{1}{2}\min\left(\gamma,\beta\right)$ and $\gamma=\beta$ of \citep[Theorem 1]{giles_2015}, which entails that the cost of the multilevel estimator $\PML^{\cR}$ is $\cO(\ln(\ve)^{2}\ve^{-2})$.
	
	For $\PML^{\cT}$, we proceed similarly. The required cost $C_{l}^{\cT}$
	at each level is still $\cO(n_{\ell}^{2})$, therefore we have $\gamma=2$.
	The case $p=2$ in Proposition \ref{prop:strong:Lp:error} implies that there exists a constant $c$ independent of $\ell$ such that 
	\begin{align*}
		V_{\ell}^{\cT} & \leq c \, 2^{-2(1+H)\ell} 
	\end{align*}
	for all $\ell \ge 1$.
	Consequently, we have $\beta=2(1+H)$ and the optimal number of samples $M_{\ell,\cT}$ at level $\ell$ is now of the form $M_{\ell,\cT} = M_{0,\cT} \, 2^{-\frac{\left(\beta+\gamma\right)}{2}\ell} = M_{0,\cT} \, 2^{-(2+H)\ell}$.
	Using again Proposition \ref{prop:strong:Lp:error} to estimate the bias  $\bE\bigl[\varphi\bigl(\VIX_{T}^{2,\cR_{L_{\cR}}}\bigr) - \varphi(\VIX_{T}^{2})\bigr]$, we obtain $\alpha=1+H$. 
	Therefore, we are in the case $\alpha\geq\frac{1}{2}\min\left(\gamma,\beta\right)$ and  $\gamma<\beta$ of \citep[Theorem 1]{giles_2015}, which entails that the complexity of $\PML^{\cT}$ is $\cO\left(\ve^{-2}\right)$. 
\end{proof}

\begin{table}[H] 
	\begin{centering}
		\begin{tabular}{|c|c|c|}
			\hline 
			Scheme  & Standard MC  & Multilevel MC\tabularnewline
			\hline 
			Rectangle  & $\cO(\ve^{-4})$  & $\cO(\ln^{2}(\ve)\ve^{-2})$\tabularnewline
			\hline 
			Trapezoidal  & $\cO\Bigl(\ve^{-2\bigl(1+\frac{1}{1+H}\bigr)}\Bigr)$  & $\cO(\ve^{-2})$\tabularnewline
			\hline 
		\end{tabular}
		\par\end{centering}
	\caption{\label{table:asymptotic_costs}
	Summary of the different computational costs for the rectangle and
		trapezoidal schemes combined with standard and multilevel Monte Carlo. The target  $\MSE$ is $\cO(\ve^{2})$.} 
\end{table}

\subsection{Numerical results} \label{s:numerics_ML}

Recall that the cost required to achieve $\MSE=\cO(\ve^{2})$ with the plain MC estimator \eqref{e:MC} is given by $\CostMC^{\cR}=\cO(\ve^{-4})$, while with the MLMC estimator \eqref{eq:def:mlmc:estimator} we have $\CostML^{\cR}=\cO\bigl(\ln(\ve)^{2}\ve^{-2}\bigr)$ when using the rectangle scheme $\cD=\cR$, and $\CostML^{\cT}=\cO(\ve^{-2})$ when using the trapezoidal scheme $\cD=\cT$.
Consequently, we expect to observe a linear dependence of the log of the quadratic error $\ln\left(\MSE\right)$ with respect to the log of the cost $\ln\left(\Cost\right)$, with a different slope  for each estimator:  we expect slope close to $-\frac 12$ for a regression of $\ln\left(\MSE\right)$ against $\ln\left(\CostMC^\cR\right)$, slope close to $-1$ for $\ln\left(\MSE\right)$ against $\ln\left(\CostML^{\cT}\right)$, and slope still close to $-1$ for $\ln\left(\MSE\right)$ against $\ln \left(\CostML^{\cR}\right)$; in the latter case, the slope will be slightly affected by the logarithmic term $\cO\bigl(\ln(\ve)^{2})$.

To compare the four different estimators, namely the plain MC and the MLMC based on the rectangle and the trapezoidal schemes,
we price two at-the-money VIX call options with parameters
$T=0.5, \, H=0.1, \, \eta=0.5, \, X_{0}=\ln\left(0.235^{2}\right)$, $\Delta \in \{ \frac{1}{12}, 1 \}$. When $\Delta = \frac{1}{12}$, 
this corresponds to a $\VIX$ call option, while when $\Delta = 1$, this corresponds to a call option on the forward variance $V_T^{T, T + 1}$ overlooking a one--year time window.
For each call option, we display the behavior of the resulting MSE against the computational cost in log-log plot, and the estimated prices against the computational cost (Figures \ref{fig:1M_mse_price_no_cv} and \ref{fig:12M_mse_price_no_cv}).
The reference prices
are computed with an intensive MC simulation
over $500$ discretization points, $2 \times 10^{7}$ i.i.d.\ samples and taking advantage of the control variate presented in section \ref{subsec:control:variate}.
For $\Delta = \frac{1}{12}$, the reference VIX future price (which sets the option strike) is $0.2218135 \pm 3 \times10^{-7}$ and the reference call option price is $0.0298840 \pm 3 \times10^{-7}$.
When $\Delta = 1$, the future price equals $0.2295989 \pm 4 \times10^{-7}$ and the call option price is $0.0196369 \pm 4 \times10^{-7}$.

For given values of $\ve$, we implement the Monte Carlo estimator $\widehat{P}_{M}^{\cR_{n}}$ in \eqref{e:MC} with $M=\left\lceil \ve^{-2}\right\rceil $ and $n=\left\lceil \ve^{-1}\right\rceil $. 
For the multilevel estimator $\PML^{\cR}$,
we set $n_{0}=6$ and, following Remark \ref{rem:multilevel:free:parameters}, we take
$L_{\cR}=\left\lceil \frac{\ln\left(\sqrt{2} \, c_{1}  \,\ve^{-1}\right)}{\ln\left(2\right)}\right\rceil$ and $M_{0,\cR}=\left\lceil \frac{2}{\ve^{2}}  \, c_{2}\left(L_{\cR}+1\right)\right\rceil$, where, according to Remark \ref{rem:lipschitz:payoff}, $\mathrm{L}_{\varphi}=\frac{1}{2\kappa}$. 
Recall that the other constants are $c_{1}=\Lip\Lambda\left(X_{0},T,\Delta,H\right)n_{0}^{-1}$ and $c_{2} = 10 \, \Lip^{2}\Lambda\left(X_{0},T,\Delta,H\right)^{2}n_{0}^{-2}$.
For the multilevel estimator $\PML^{\cT}$ with the trapezoidal scheme, we use the same parameters $M_{0,\cT}=M_{0,\cR}$, $L_{\cT}=L_{\cR}$, and we set $n_{0}=3$.
The simulation effort $M_\ell$ at each level is then fixed according to Theorem \ref{thm:mlmc:discetization:integral}. 

We estimate the $\MSE$ for each method as $\frac{1}{N_{\MSE}}\sum_{j=1}^{N_{\MSE}}\left(\hat{p}_{j}-p\right)^2$, where $p$ is the reference price and $\left(\hat{p}_{j}\right)_{1\leq j\leq N_{\MSE}}$ are $N_{\MSE}$ independent copies of either the multilevel or the plain MC estimators, along with their $95\%$-confidence interval (materialized by the vertical error bars 
in Figures \ref{fig:1M_mse_price_no_cv} and \ref{fig:12M_mse_price_no_cv}).
We take $N_{\MSE}=10^3$.
The cost of each estimator is evaluated via the identities $\CostMC^\cR=n^{2}M$,
$\CostML^{\cR} = \sum_{\ell=0}^{L_{\cR}}n_{\ell}^{2} \, M_{\ell,\cR} = n_{0}^{2} \, M_{0,\cR}\left(L_{\cR}+1\right)$,
and $\CostML^{\cT} = n_{0}^{2} \, M_{0,\cT}\sum_{\ell=0}^{L_{\cT}}2^{-H\ell} = n_{0}^{2} \, M_{0,\cT} \, \frac{1-2^{-H(L_{\cT}+1)}}{1-2^{-H}}.$

\begin{figure}[t]
	\begin{centering}
		\includegraphics[width=7.9cm]{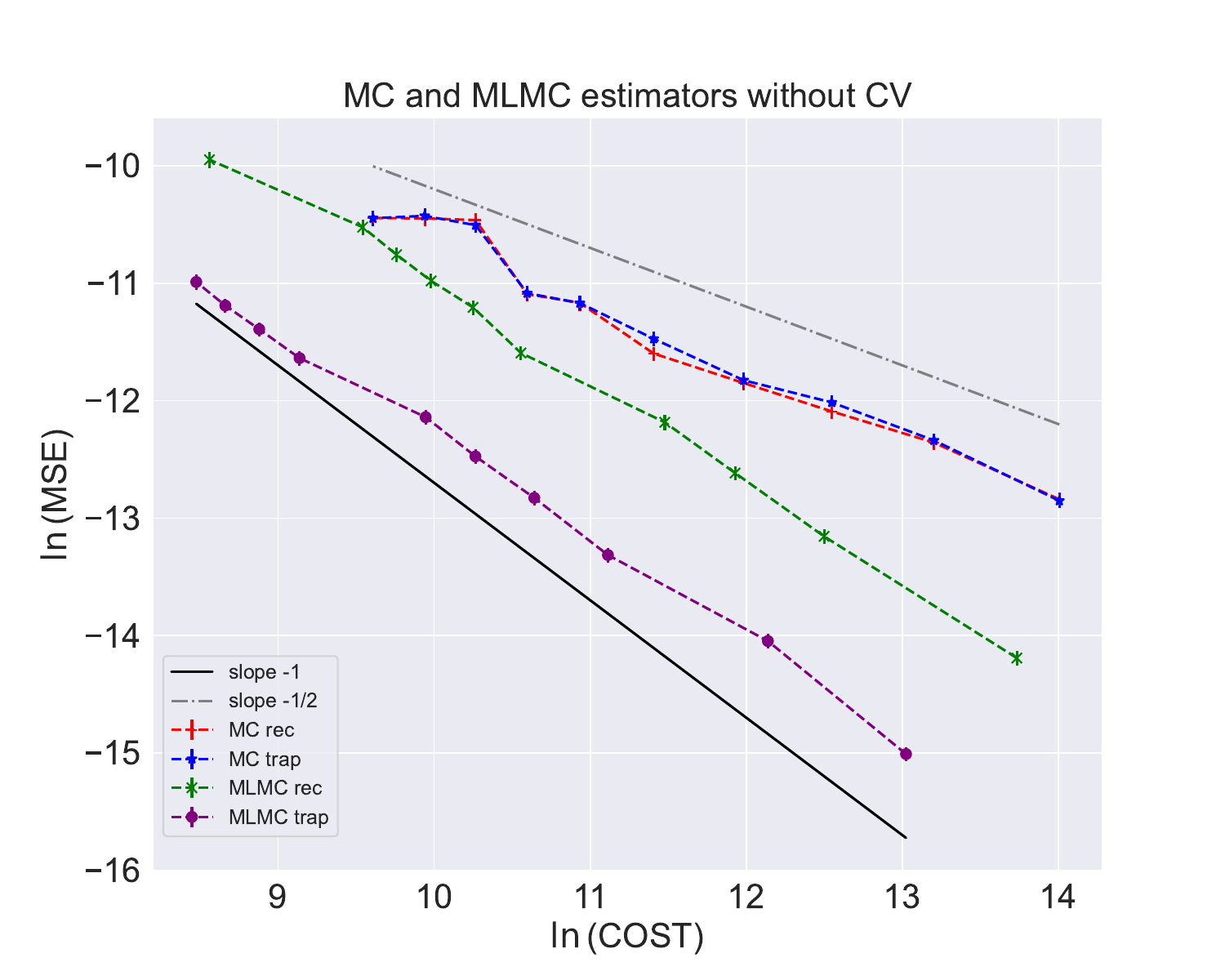}
		\includegraphics[width=7.9cm]{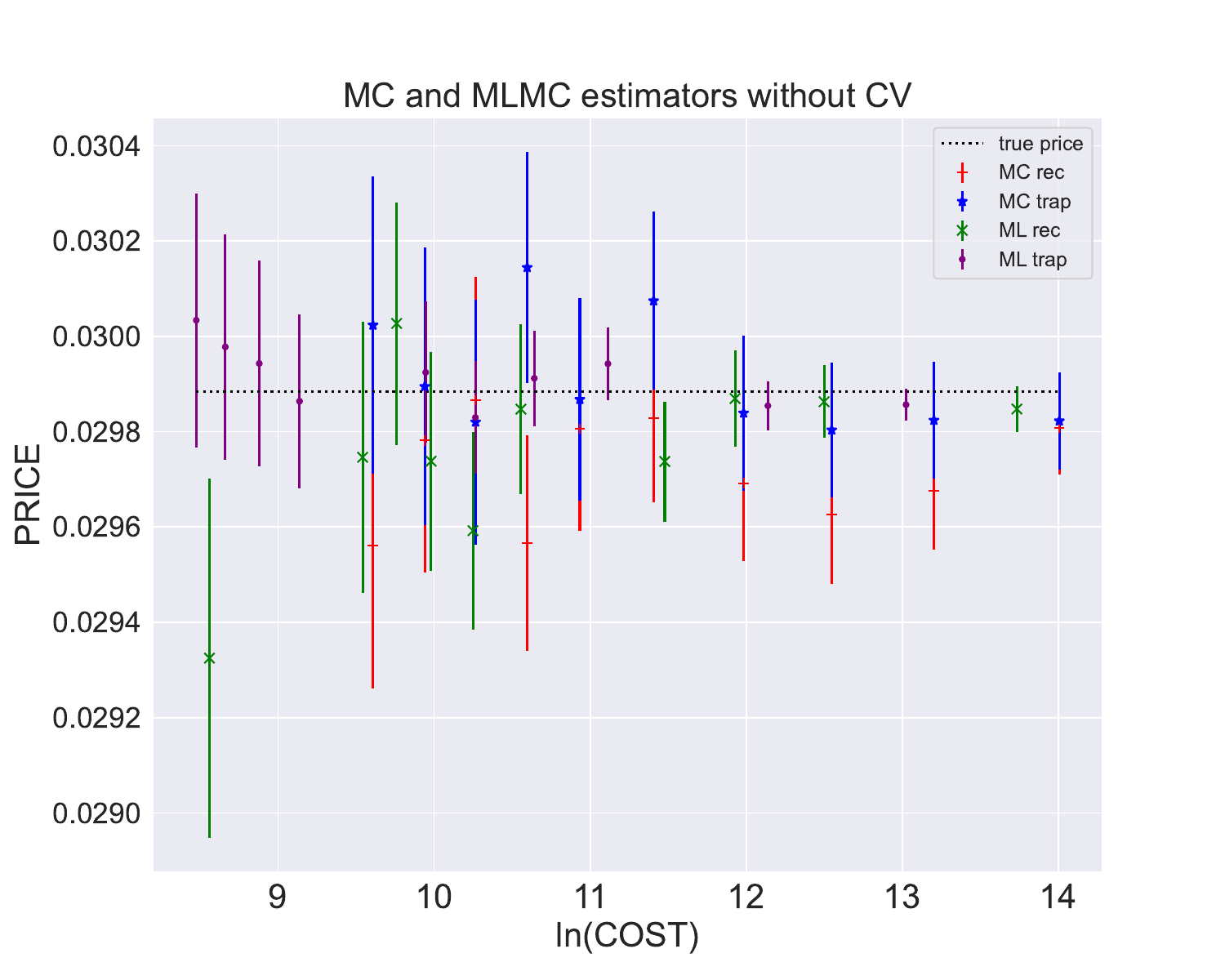} 
		\par\end{centering}
	\caption{
	Pricing of a $\VIX$ call option, $\Delta = \frac{1}{12}$.
		\label{fig:1M_mse_price_no_cv}
		\emph{Left}: log of the mean squared error $\ln\left(\MSE\right)$ against the log of the computational cost $\ln\left(\Cost\right)$ for the MC \eqref{e:MC} and the MLMC \eqref{eq:def:mlmc:estimator} estimators based on the rectangle and trapezoidal schemes, without control variate. \emph{Right}: the associated VIX call option prices against $\ln\left(\Cost\right)$; we display the empirical means of the estimators over the $N_{\MSE}$ independent runs, along with their $95\%$ confidence interval.
	}
\end{figure}

\begin{figure}[t]
	\begin{centering}
		\includegraphics[width=7.9cm]{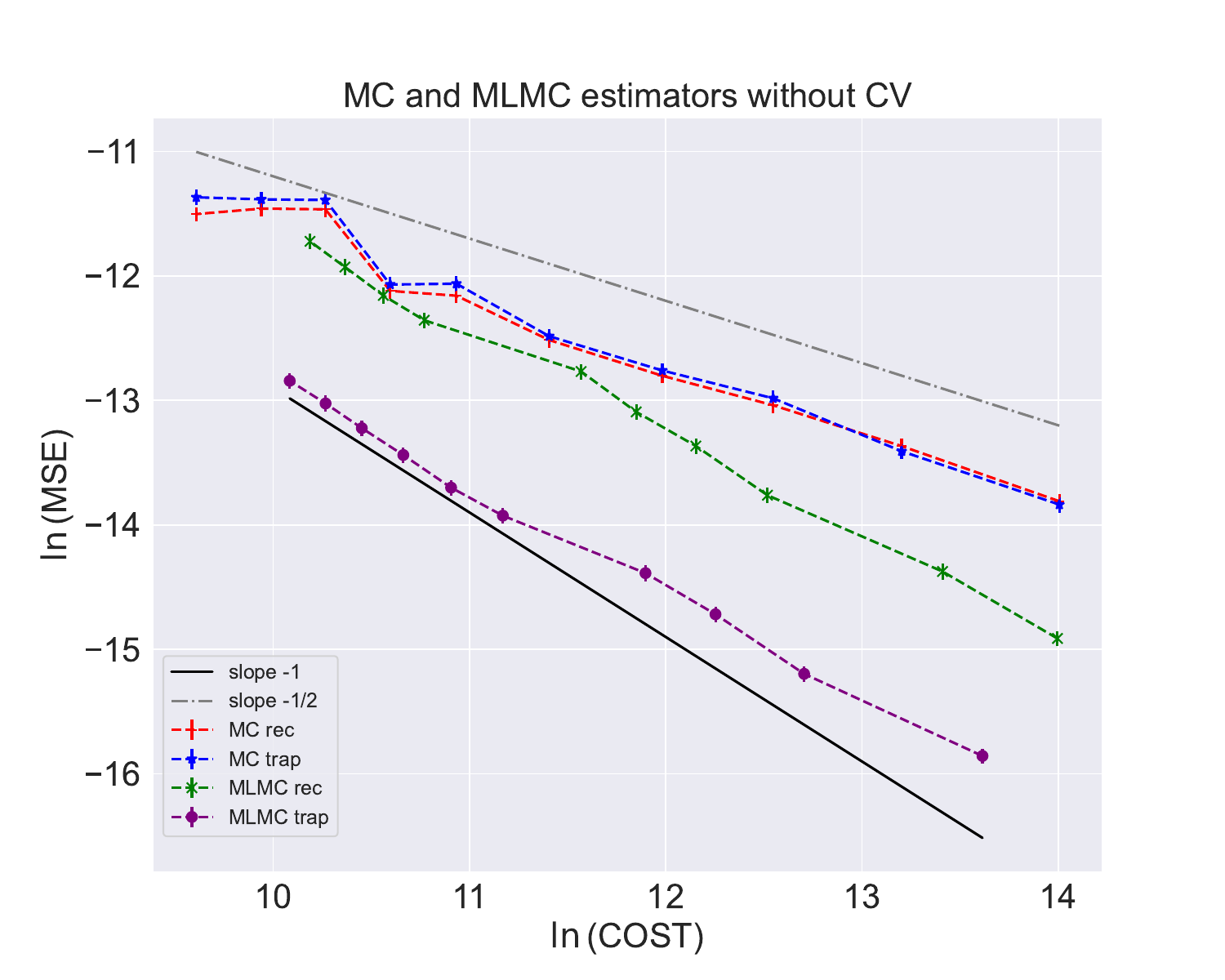}
		\includegraphics[width=7.9cm]{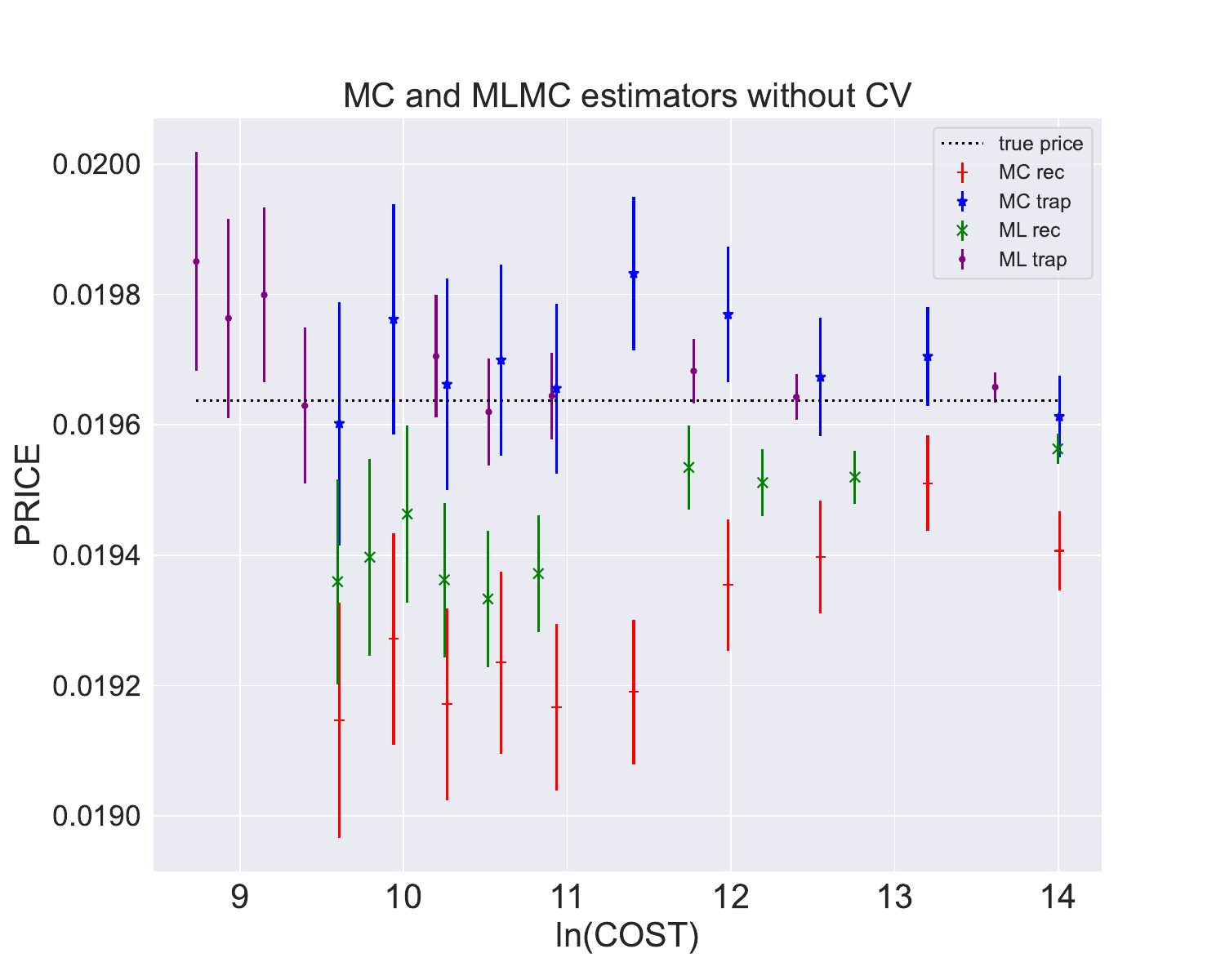} 
		\par\end{centering}
	\caption{
		\label{fig:12M_mse_price_no_cv}
		Pricing of a call option on 1Y forward variance, $\Delta = 1$.
		\emph{Left}: log of the mean squared error $\ln\left(\MSE\right)$ against 
		the log of the computational cost $\ln\left(\Cost\right)$ for the MC \eqref{e:MC} and the 
		MLMC \eqref{eq:def:mlmc:estimator} 
		estimators based on the rectangle and trapezoidal schemes,
		without control variate. \emph{Right}:	the associated call option prices against $\ln\left(\Cost\right)$; 
		we display the empirical means of the estimators over the $N_{\MSE}$ independent runs, along with their $95\%$ confidence interval.
	}
\end{figure}

In the left graphs of Figures \ref{fig:1M_mse_price_no_cv} and \ref{fig:12M_mse_price_no_cv}, we retrieve
the expected asymptotic slopes.
From our theoretical analysis, see Table \ref{table:asymptotic_costs}, we expect a slope close to $-\frac 12$
for the plain MC estimator with the rectangle scheme, and a slope close to $-\frac{1+H}{2+H}$ for the MC estimator with the trapezoidal scheme. Since $\frac{1+H}{2+H} = \frac 12 + \frac H 4 + \cO(H^2)$, for small values of $H$ the two slopes are very close, in line with what is observed in Figures \ref{fig:1M_mse_price_no_cv} and \ref{fig:12M_mse_price_no_cv} (recall we chose $H=0.1$). For the multilevel method, we observe the expected slope close to $-1$ for both estimators (actually slightly less than $-1$ for the multilevel method with a rectangle scheme, due to the logarithmic term in $\cO(\ve^{-2} \ln^2(\ve))$).
One can also see that, for a given cost, the estimated $\MSE$ is smaller for the multilevel estimators than  for the plain MC estimators.

In the right graphs of Figures \ref{fig:1M_mse_price_no_cv} and \ref{fig:12M_mse_price_no_cv}, 
we display the empirical mean $\frac 1 {N_{\MSE}} \sum_{j=1}^{N_{\MSE}} \hat{p}_{j}$ of each estimator over the  $N_{\MSE}$ independent runs, along with the $95\%$-confidence interval for the empirical mean. As expected, we see that the MLMC estimator with a trapezoidal scheme performs best, that is, 
it has the smallest bias and variance compared to the other three estimators. The bias is smaller for the MC estimator with a trapezoidal scheme compared to that with a rectangle scheme.
For large costs, while the bias are close, the variance of the MLMC estimator with a rectangle scheme is smaller than that of the two MC estimators; 
the resulting MSE for the MLMC estimator is thus smaller than in the pure MC case (as seen in the left graphs of Figures \ref{fig:1M_mse_price_no_cv} and \ref{fig:12M_mse_price_no_cv}).

\subsubsection{Additional contribution of the control variate}

In practice, options on forward variance can be priced exploiting the supplementary contribution of the control variate 
presented in section \ref{subsec:control:variate}.
In order to assess numerically the additional error reduction, we price the two at-the-money call options considered in the previous section combining the MC and MLMC estimators 
with the control variate.
The controlled MC estimator is given by \eqref{e:option_price_with_cv} with $\varphi\left(x\right)=\left(\sqrt{x}-\kappa\right)_{+}$.
In order to build enhanced MLMC estimators, we inject the control variates \eqref{e:option_price_with_cv} and \eqref{e:CV_trap} at each level $0 \leq \ell \leq L$ of the corresponding multilevel scheme \eqref{eq:def:mlmc:estimator}.
The results are presented in Figure \ref{fig:cv}.
First, comparing with the right graphs of Figures \ref{fig:1M_mse_price_no_cv} and \ref{fig:12M_mse_price_no_cv}, we observe that the control variate significantly reduces the variance of each estimator, as expected, so that the bias can be identified more clearly.
Second, one can see that the MLMC estimator with a trapezoidal scheme still provides the best results, while the MC estimator with a rectangle scheme is still the worst.

\begin{figure}[t]
		\begin{centering}
			\includegraphics[width=7.7cm]{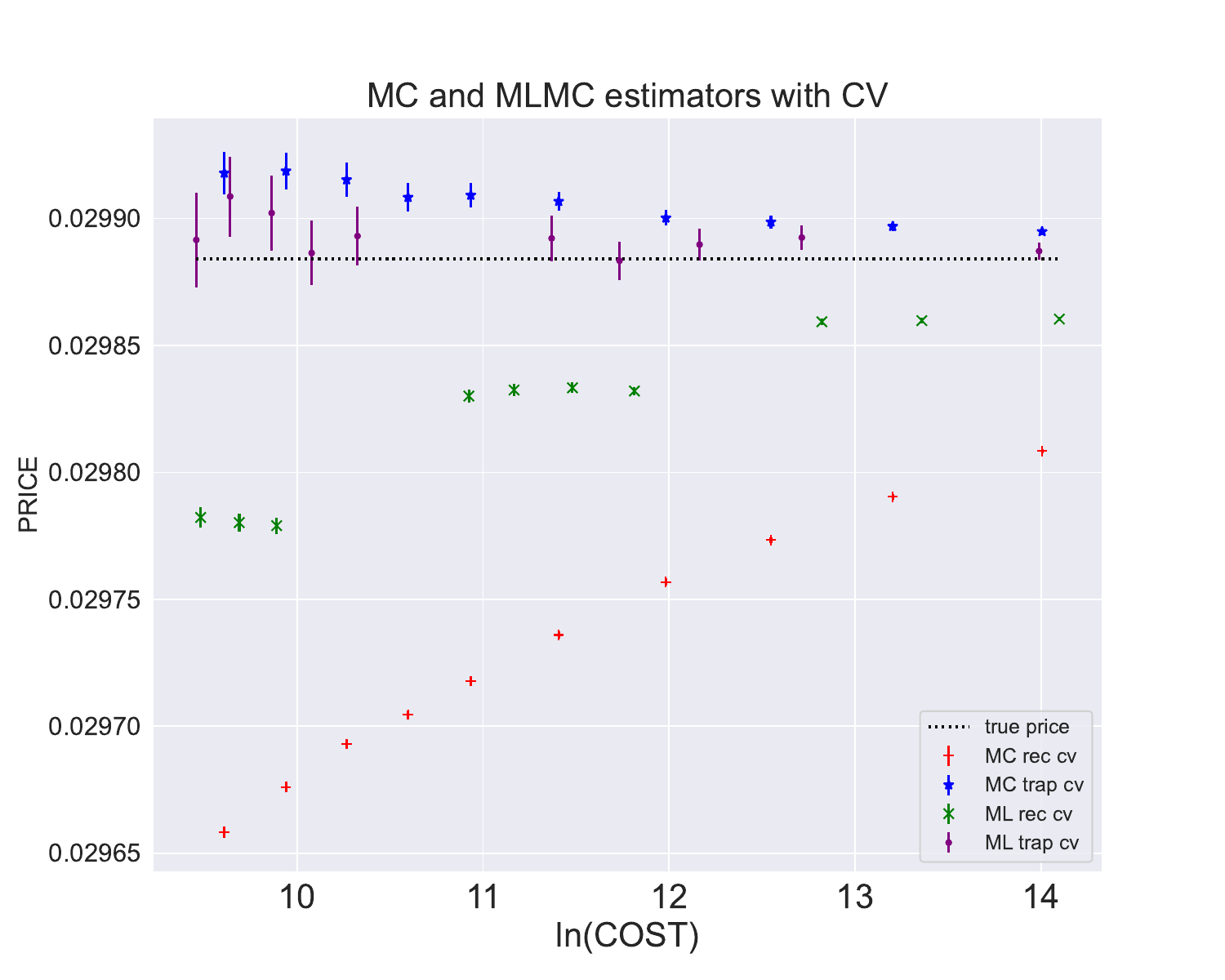}
			\includegraphics[width=7.7cm]{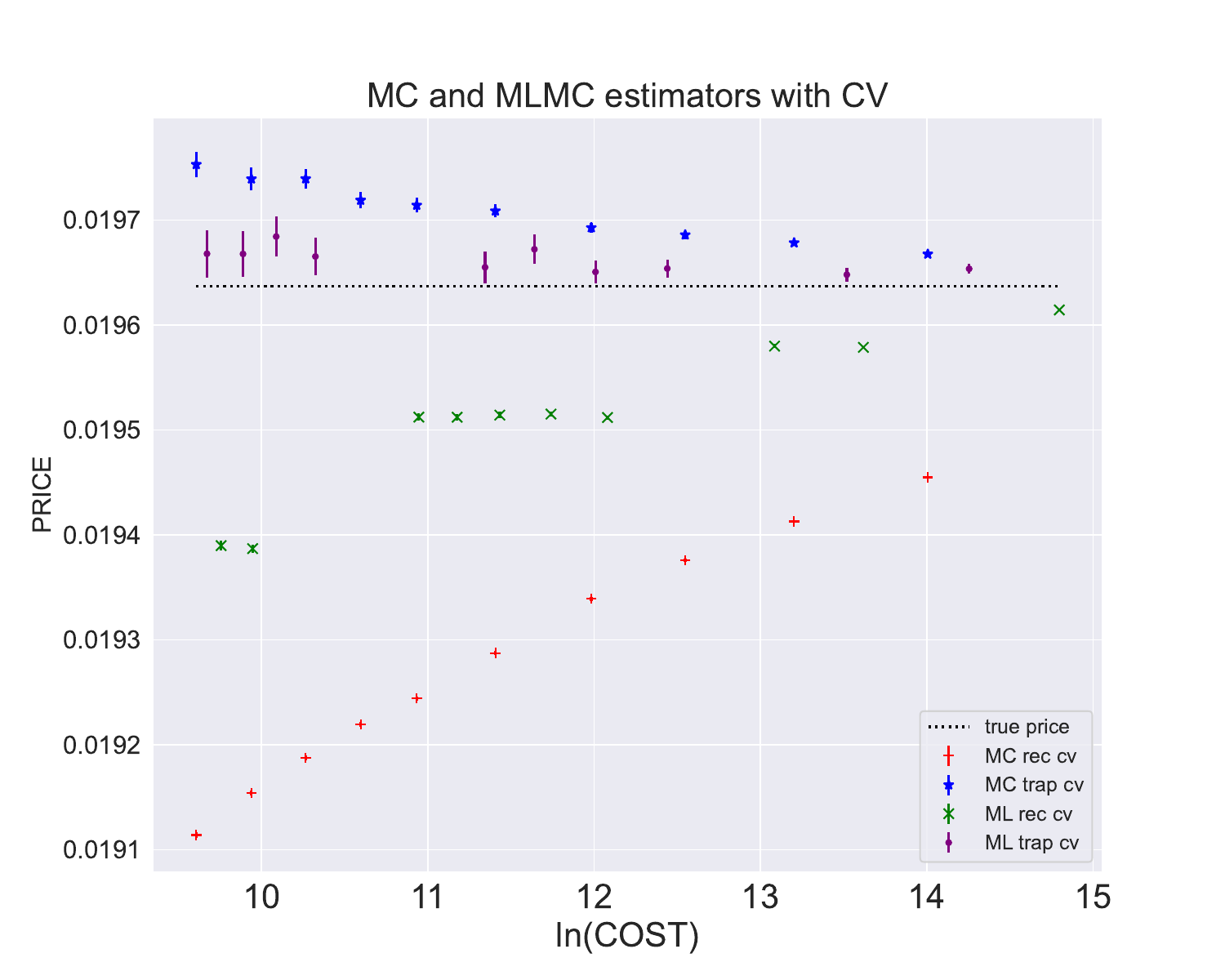} 
			\par\end{centering}
		\caption{
			\label{fig:cv}
			Prices of call options on forward variance against $\ln\left(\Cost\right)$ for $\Delta = \frac{1}{12}$ (VIX case, left figure) and for $\Delta = 1$ (right figure) for the MC and MLMC estimators based on the different discretization scheme and enhanced with the control variates presented in section \ref{subsec:control:variate}.
}
\end{figure}

\section{Conclusion}

We have presented two multilevel Monte Carlo estimators for the evaluation of $\VIX$ options in rough forward variance models, based on a rectangle and a trapezoidal discretization scheme for the integrated forward variance defining the $\VIX$ random variable.
We have shown that the trapezoidal scheme combined with the multilevel method is able to achieve the optimal complexity $\cO(\ve^{-2})$ of an unbiased estimator with mean-squared error of size $\ve^2$.

\section{Appendix} \label{s:appendix}

\begin{proof}[Proof of Proposition \ref{prop:strong:Lp:error}]
	We follow the steps of the proof of \citep[Proposition 2]{horvath2018volatility}.
	We denote $C$ a constant that may vary from line to line and may depend on the model parameters, on $T$ and $\Delta$ and on the exponent $p$, but not on $n$.
	Let us consider the rectangle scheme. 
	Without loss of generality, assume $p\geq1$, and define 
	\[
	Z_{u} = -\frac{1}{2}\int_{0}^{T}K(u,s)^{2}\dd s+\int_{0}^{T}K(u,s)\dd W_{s},
	\qquad
	\eta(u)=
	\min\left\{ u_{i}: u_{i} \ge u\right\} 
	\]
	Applying Minkowski inequality, one has
	\begin{align*}
		\left\Vert \VIX_{T}^{2}-\Rn\right\Vert _{p}
		& =
		\frac{1}{\Delta}\left\Vert \int_{T}^{T+\Delta}e^{X_{0}^{u}}\left(e^{Z_{u}}-e^{Z_{\eta\left(u\right)}}\right)\dd u\right\Vert _{p}
		\\
		& \leq\frac{1}{\Delta}\int_{T}^{T+\Delta}e^{X_{0}^{u}}\left\Vert e^{Z_{u}}-e^{Z_{\eta(u)}}\right\Vert _{p}\dd u
		=  \frac{1}{\Delta} \sum_{i=1}^n \int_{u_{i-1}}^{u_i} e^{X_{0}^{u}} \left\Vert e^{Z_{u}} - e^{Z_{u_i}}\right\Vert _{p}\dd u \,.
	\end{align*}
	For every  $u\in(u_{i-1},u_{i})$ there exists $\theta\in[0,1]$ (generally random)
	such that 
	\begin{align*}
		\left\Vert e^{Z_{u}}-e^{Z_{u_{i}}}\right\Vert _{p} & =\left\Vert \left(Z_{u}-Z_{u_{i}}\right)e^{Z_{u_{i}}+\theta\left(Z_{u}-Z_{u_{i}}\right)}\right\Vert _{p}\\
		& \leq\left\Vert Z_{u}-Z_{u_{i}}\right\Vert _{2p}\left\Vert e^{M}\right\Vert _{2p}
		\\
		& \leq C\left\Vert \int_{0}^{T}\left(K(u,s)-K(u_{i},s)\right)\dd W_{s}\right\Vert _{2p}+\left|\int_{0}^{T}\left(K(u,s)^{2}-K(u_{i},s)^{2}\right)\dd s\right|,
	\end{align*}
	where we have applied Cauchy--Schwarz and
	Minkowski inequalities, and 
	%$C>0$ is an independent constant
	we have used the fact that the supremum $M := \sup_{T\leq u\le T+\Delta}Z_{u}$ of the Gaussian process $Z_u$ admits finite exponential moments, so that we can set $C = \left\Vert e^{M}\right\Vert _{2p}$.
	Now, since
	\[
	\int_{0}^{T}\left(K(u,s)-K(u_{i},s)\right)\dd W_{s}\egl\cN\left(0,\int_{0}^{T}\left(K(u,s)-K(u_{i},s)\right)^{2}\dd s\right),
	\]
	using the moment scaling property $\bE\left[X^{2p}\right] = c_p \,
	%\frac{(2p)!}{2^{p}p!}
	\bE\left[X^{2}\right]^p$ of Gaussian random variables, we obtain
	\[
	\left\Vert e^{Z_{u}}-e^{Z_{u_{i}}}\right\Vert _{p}
	\leq
	C \left(\int_{0}^{T}\bigl(K(u,s)-K(u_{i},s)\bigr)^{2}\dd s\right)^{1/2}
	+
	\left|\int_{0}^{T}\left(K(u,s)^{2}-K(u_{i},s)^{2}\right)\dd s\right| \,.
	\]
	%for another constant $C>0$.
	The key ingredient is that the integration kernel $K(u,s) = \eta (u-s)^{H - \frac{1}{2}}$ of the rough Bergomi model satisfies the following condition (see \citep[Proof of Corollary 1]{horvath2018volatility}): there exists a constant $c>0$ such that 
	\be \label{e:condition_kernel_1}
	\left(\int_{0}^{T}
	\bigl( K(u_{1},s) - K(u_{2},s) \bigr)^{2}\dd s
	\right)^{1/2}
	\leq c(u_{2}-u_{1})(u_{2}-T)^{H-1},\qquad\text{for all }T\leq u_{1}<u_{2} \,.
	\ee
	On the other hand, by the Cauchy--Schwarz inequality we also have, for all $T\leq u_{1}<u_{2} \le T + \Delta$
	\[
	\begin{aligned}
		\biggl|\int_{0}^{T}\bigl(K(u_1,s)^{2} &- K(u_2,s)^{2}\bigr)\dd s \biggr| =
		\biggl|
		\int_{0}^{T}\bigl( K(u_{1},s) + K(u_{2},s) \bigr)\bigl( K(u_{1},s) - K(u_{2},s) \bigr)\dd s
		\biggl|
		\\
		&\le
		\biggl( \int_{0}^{T} \bigl( K(u_{1},s) + K(u_{2},s) \bigr)^{2} \dd s \biggr)^{1/2}
		\biggl( \int_{0}^{T} \bigl( K(u_{1},s) - K(u_{2},s) \bigr)^{2} \dd s \biggr)^{1/2}
		\\
		&\le
		c_\Delta (u_{2}-u_{1})(u_{2}-T)^{H-1} \,,
	\end{aligned}
	\]
	for some constant $c_\Delta$, possibly depending on $\Delta$ as well.
	Putting things together and using the local boundedness assumption on the function $u \mapsto X_0^u$, we obtain
	\be \label{e:final_estimate_rect}
	\begin{aligned}
		\left\Vert \VIX_{T}^{2}-\Rn\right\Vert _{p}
		&\le
		C \sum_{i=1}^n
		\int_{u_{i-1}}^{u_i}  (u_i - u) (u_i - T)^{H-1}
		\dd u
		\\
		&\le
		C \sum_{i=1}^n  (u_i - u_{i-1})^2 (u_i - T)^{H-1}
		=
		C \frac{\Delta^{1+H}}{n^{1+H}}
		\sum_{i=1}^n \frac 1{i^{1 - H}}
		= \cO\left(\frac{1}{n}\right).
	\end{aligned}
	\ee
	
	For the trapezoidal scheme, the proof follows the same steps.
	Instead of \eqref{e:condition_kernel_1}, we use the following property (see again \citep[Proof of Corollary 1]{horvath2018volatility}): for all $T\leq u_{1}\leq u_{2}<u_{3}$, 
	\[
	\biggl(\int_{0}^{T}
	\Bigl(K(u_{2},s)-\frac{u_{3}-u_{2}}{u_{3}-u_{1}}K(u_{1},s)-\frac{u_{2}-u_{1}}{u_{3}-u_{1}}K(u_{3},s)\Bigr)^{2}\dd s \biggr)^{\frac{1}{2}}
	\leq c(u_{3}-u_{1})^{2}(u_{3}-T)^{H-2} \,.
	\]
	The estimate \eqref{e:final_estimate_rect} is then replaced by
	\[
	\begin{aligned}
		\left\Vert \VIX_{T}^{2}-\Tn\right\Vert _{p}
		&\le
		C \sum_{i=1}^n  (u_i - u_{i-1})^3 (u_i - T)^{H-2}
		=
		C \frac{\Delta^{1+H}}{n^{1+H}}
		\sum_{i=1}^n \frac 1{i^{2 - H}}
		= \cO\left(\frac{1}{n^{1+H}}\right).
	\end{aligned}
	\]
\end{proof}

\begin{proof}[Proof of Theorem \ref{thm:strong:L2:error}]
	The squared $L^{2}$ strong error can be rewritten as 
	\begin{align*}
		\mathbb{E}\left[\left|\VIX_{T}^{2}-\Rn\right|^{2}\right] & =\mathbb{E}\left[\left(\frac{1}{\Delta}\sum_{i=1}^{n}\int_{u_{i-1}}^{u_{i}}\left(e^{X_{T}^{t}}-e^{X_{T}^{u_{i}}}\right)\dd t\right)^{2}\right]\\
		& =\frac{1}{\Delta^{2}}\mathbb{E}\left[\sum_{i,j=1}^{n}\int_{u_{i-1}}^{u_{i}}\left(e^{X_{T}^{t}}-e^{X_{T}^{u_{i}}}\right)\dd t\int_{u_{j-1}}^{u_{j}}\left(e^{X_{T}^{s}}-e^{X_{T}^{u_{j}}}\right)\dd s\right]\\
		& =\frac{h^{2}}{\Delta^{2}}\sum_{i,j=1}^{n}\int_{0}^{1}\int_{0}^{1}\mathbb{E}\left[\left(e^{X_{T}^{u_{i}-uh}}-e^{X_{T}^{u_{i}}}\right)\left(e^{X_{T}^{u_{j}-vh}}-e^{X_{T}^{u_{j}}}\right)\right]\dd u \, \dd v,
	\end{align*}
	where we have introduced the new variables $u=\frac{u_{i}-t}{h},\,v=\frac{u_{j}-s}{h}$.
	Using the expression of the moment-generating function of the normal distribution, for all $a,b\in\left[T,T+\Delta\right]$ we have 
	\begin{align*}
		\mathbb{E}\left[e^{X_{T}^{a}+X_{T}^{b}}\right] & =e^{X_{0}^{a}+X_{0}^{b}} \, \mathbb{E}\left[e^{-\frac{1}{2}\int_{0}^{T}\left(K\left(a,t\right)^{2}+K\left(b,t\right)^{2}\right)\dd t+\int_{0}^{T}\left(K\left(a,t\right)+K\left(b,t\right)\right)\dd W_{t}}\right]
		\\
		& = e^{X_{0}^{a}+X_{0}^{b}} \, e^{\int_{0}^{T}K\left(a,t\right)K\left(b,t\right)\dd t} \, ;
	\end{align*}
	therefore,
	%the expectation inside the double integral simplifies to 
	\begin{align*}
		& \mathbb{E}\left[\left(e^{X_{T}^{u_{i}-uh}}-e^{X_{T}^{u_{i}}}\right)\left(e^{X_{T}^{u_{j}-vh}}-e^{X_{T}^{u_{j}}}\right)\right]\\
		& \qquad\qquad=e^{X_{0}^{u_{i}-uh}+X_{0}^{u_{j}-vh}}e^{\int_{0}^{T}K\left(u_{i}-uh,t\right)K\left(u_{j}-vh,t\right)\dd t}-e^{X_{0}^{u_{i}-uh}+X_{0}^{u_{j}}}e^{\int_{0}^{T}K\left(u_{i}-uh,t\right)K\left(u_{j},t\right)\dd t}\\
		& \qquad\qquad\qquad - e^{X_{0}^{u_{i}}+X_{0}^{u_{j}-vh}}e^{\int_{0}^{T}K\left(u_{i},t\right)K\left(u_{j}-vh,t\right)\dd t}+e^{X_{0}^{u_{i}}+X_{0}^{u_{j}}}e^{\int_{0}^{T}K\left(u_{i},t\right)K\left(u_{j},t\right)\dd t}.
	\end{align*}
	Using the assumption $X_{0}^{u}=X_{0}$ for all $u\in[T,T+\Delta]$, it follows that
	\be
	\begin{aligned}
		\delta(n) :=&
		\frac{n^{2}}{e^{2X_{0}}} \mathbb{E}\Bigl[\bigl|\VIX_{T}^{2}-\Rn\bigr|^{2}\Bigr]
		\\
		=&
		\int_{0}^{1}\int_{0}^{1}\sum_{i,j=1}^{n}\Bigl(e^{\int_{0}^{T}K\left(u_{i}-uh,t\right)K\left(u_{j}-vh,t\right)\dd t}-e^{\int_{0}^{T}K\left(u_{i}-uh,t\right)K\left(u_{j},t\right)\dd t}
		\\
		& \qquad\qquad\qquad -e^{\int_{0}^{T}K\left(u_{i},t\right)K\left(u_{j}-vh,t\right)\dd t}+e^{\int_{0}^{T}K\left(u_{i},t\right)K\left(u_{j},t\right)\dd t}\Bigr)\dd u\,\dd v
		\\
		= & 
		\int_{0}^{1}\int_{0}^{1}\sum_{i,j=1}^{n}\Bigl(e^{g\left((i-u)h,\,(j-v)h\right)}-e^{g\left((i-u)h,\,jh\right)}-e^{g\left(ih,\,(j-v)h\right)}+e^{g\left(ih,\,jh\right)}\Bigr)\dd u\,\dd v
		\\
		= & 
		\int_{0}^{1} \dd u \int_{0}^{1} \dd v\sum_{i,j=1}^{n}
		\int_{ih - uh}^{ih} \dd x \int_{jh - vh}^{jh} \dd y \,
		F(x,y) \,,
	\end{aligned}
	\label{eq:proof:sum:exp}
	\ee
	where we have used $u_i = T + i \, h$ and we have set
	%made the change of variables $s=T-t$ 
	\[
	\begin{aligned}
		&F(x,y) := \partial_{x y} e^{g(x,y)}
		\qquad \forall \, x, y \maj 0;
		\\
		&g(x,y) := \eta^{2}\int_{0}^{T}\left(s+x\right)^{H-\frac{1}{2}}\left(s+y\right)^{H-\frac{1}{2}}\dd s = g(y,x),
		\qquad \forall \, x, y \ge 0.
	\end{aligned}
	\]
	By construction, the function $g$ is symmetric, $F$ is positive and integrable over $(0,\Delta)^2$, and 
	\[
	\begin{aligned}
		\int_{0}^{\Delta}\int_{0}^{\Delta}F(x,y)\dd x\dd y
		&= e^{g\left(\Delta,\Delta\right)}-2e^{g\left(\Delta,0\right)}+e^{g\left(0,0\right)}
		\\
		&= e^{\eta^{2}\frac{\left(T+\Delta\right)^{2H}-\Delta^{2H}}{2H}}-2e^{\eta^{2}\int_{0}^{T}t^{H-\frac{1}{2}}\left(t+\Delta\right)^{H-\frac{1}{2}}\dd t}+e^{\eta^{2}\frac{T^{2H}}{2H}} =: \lambda_\Delta.
	\end{aligned}
	\]
	We can simplify the expression of $\delta(n)$ in \eqref{eq:proof:sum:exp} into
	\begin{align}
		\delta(n) &
		%=\int_{0}^{1}\dd u\int_{0}^{1}\dd v\sum_{i,j=1}^{n}\int_{ih-uh}^{ih}\dd x\int_{jh-vh}^{jh}\dd y \, F(x,y)
		=\sum_{i,j=1}^{n}\int_{(i-1)h}^{ih}\dd x\int_{(j-1)h}^{jh}\dd y \, F(x,y)\int_{\frac{ih-x}{h}}^{1}\dd u\int_{\frac{jh-y}{h}}^{1}\dd v \nonumber
		\\
		&=
		\sum_{i,j=1}^{n}\int_{(i-1)h}^{ih}\dd x\int_{(j-1)h}^{jh}\dd y \, F(x,y) \Bigl(1-\frac{ih-x}{h}\Bigr)\Bigl(1-\frac{jh-y}{h}\Bigr). \label{eq:proof:delta:n:representation:ih:jh}
		%\\
		%&=h^{2}\int_{0}^{1}\int_{0}^{1}(1-z_{1})(1-z_{2})\sum_{i,j=1}^{n}F\left((i-z_{1})h,(j-z_{2})h\right)\dd z_{1}\dd z_{2},
	\end{align}
	%where we have made the change of variables $z_{1}=\frac{ih-x}{h}$, $z_{2}=\frac{jh-y}{h}$.
	We are going to show that there exist $\overline \delta(n)$ and $\underline \delta(n)$ such that
	$\underline \delta(n) \le \delta(n) \le \overline \delta(n)$ for all $n$, and 
	\be \label{e:encadrement}
	\underline \delta(n)
	\xrightarrow[n\to \infty]{} \frac 14  \lambda_\Delta
	\qquad \qquad
	\overline \delta(n)
	\xrightarrow[n\to \infty]{} \frac 14  \lambda_\Delta,
	%\qquad \qquad
	%\mbox{as } n \to \infty
	\ee
	so that $\delta(n) \to \frac 14  \lambda_\Delta$, too.
	Together with \eqref{eq:proof:sum:exp}, the latter limit implies $\mathbb{E}\Bigl[\bigl|\VIX_{T}^{2}-\Rn\bigr|^{2}\Bigr]^{1/2} \sim \frac{e^{X_{0}}}n \sqrt{\delta(n)} \sim \frac{e^{X_{0}}}{2n} \sqrt{\lambda_\Delta} = \frac 1 n \Lambda(X_0, T, \Delta, H)$ as $n\to \infty$, which will conclude the proof.
	
	In order to show \eqref{e:encadrement}, we exploit	some information on the first derivatives of $F$.
	Since $F(x,y)=[\partial_{x}g(x,y)\partial_{y}g(x,y)+ $ $\partial_{xy}g(x,y)] e^{g(x,y)}$, we have
	\[
	\partial_{x}F(x,y)
	= \Bigl[\left(\partial_{x}g\right)^{2}\partial_{y}g + 2 \, \partial_{x}g \, \partial_{xy}g + \partial_{xxy}g
	+ \partial_{xx}g \, \partial_{y}g \Bigr](x,y) \, e^{g(x,y)}
	\]
	for every $x,y>0$,  where
	\begin{align*}
		\partial_{x}g(x,y)&= \eta^2 \Bigl(H-\frac{1}{2}\Bigr) \int_{0}^{T} (s+x)^{H- \frac{3}{2}} (s+y)^{H-\frac{1}{2}} \dd s <0,
		\\
		\partial_{y}g(x,y)&= \eta^2 \Bigl(H-\frac{1}{2}\Bigr)  \int_{0}^{T} (s+x)^{H- \frac{1}{2}} (s+y)^{H-\frac{3}{2}} \dd s<0, 
		\\
		\partial_{xy}g(x,y)&=\eta^{2}\Bigl(H-\frac{1}{2}\Bigr)^{2}\int_{0}^{T}(s+x)^{H-\frac{3}{2}}(s+y)^{H-\frac{3}{2}}\dd s>0,
		\\
		\partial_{xx}g(x,y)&=\eta^{2}\Bigl(H-\frac{1}{2}\Bigr) \Bigl(H-\frac{3}{2}\Bigr)\int_{0}^{T}(s+x)^{H-\frac{5}{2}}(s+y)^{H-\frac{1}{2}}\dd s>0,
		\\
		\partial_{xxy}g(x,y)&=\eta^{2}\Bigl(H-\frac{1}{2}\Bigr)^{2}\Bigl(H-\frac{3}{2}\Bigr)\int_{0}^{T}(s+x)^{H-\frac{5}{2}}(s+y)^{H-\frac{3}{2}}\dd s<0 \,.
	\end{align*}
	Since $H\in(0,\frac{1}{2})$, we infer that $\partial_{x}F(x,y)<0$. Analogous computations lead to $\partial_{y}F(x,y)<0$, and
	it follows that
	\begin{equation}\label{eq:proof:ineq:partial:F}
		F(x_1,y_1)>  F(x_2,y_2)
	\end{equation} 
	for every $0<x_1 < x_2$ and $0<y_1 < y_2$.
	
	\emph{Step 1 (lower bound $\underline \delta(n)$)}.
	Setting $z_{1}=\frac{ih-x}{h}$, $z_{2}=\frac{jh-y}{h}$ %in the last line of \eqref{eq:proof:sum:exp}, 
	in \eqref{eq:proof:delta:n:representation:ih:jh},
	we get
	\[
	\begin{aligned}
		\delta(n)
		&= h^{2} \sum_{i,j=1}^{n} \int_{0}^{1}\int_{0}^{1}(1-z_{1})(1-z_{2})F\left((i-z_{1})h,(j-z_{2})h\right)\dd z_{1}\dd z_{2}
		\\
		&\ge h^{2} \sum_{i,j=1}^{n} F\left(i h,jh\right) \int_{0}^{1}\int_{0}^{1}(1-z_{1})(1-z_{2})\dd z_{1}\dd z_{2}
		\\
		&= \frac{h^{2}} 4 \sum_{i,j=1}^{n} F\left(i h,jh\right) =: \underline \delta(n),
	\end{aligned}
	\]
	where we have used $F\left((i-z_{1})h,(j-z_{2})h\right) \ge F\left(i h,jh\right)$
	for every $(z_1,z_2) \in [0,1]^2$
	%\eqref{eq:proof:ineq:partial:F}
	in the second step and $\int_{0}^{1}\int_{0}^{1}(1-z_{1})(1-z_{2})\dd z_{1}\dd z_{2}=\frac{1}{4}$ in the last step.
	Since $F$ is continuous and integrable over $(0,\Delta)^2$, we have $\underline \delta(n) \to \frac 14 \int_{0}^{\Delta}\int_{0}^{\Delta}F(x,y)\dd x\dd y = \frac 14 \lambda_\Delta$, as claimed.
	
	\emph{Step 1 (upper bound $\overline \delta(n)$)}.
	Separating the contribution for $i=1$, the one for $j=1$, 
	and the remaining one for $(i,j) \in \{2, \dots, n\}^2$ inside \eqref{eq:proof:delta:n:representation:ih:jh}, we get
	\[
	\begin{aligned}
		\delta(n)
		&= 
		\int_{0}^{h} \dd x \, \frac x h \, \sum_{j=1}^{n} \int_{(j-1)h}^{jh}\dd y \, F(x,y)  \Bigl(1-\frac{jh-y}{h}\Bigr)
		\\
		&\quad +
		\int_{0}^{h} \dd y \, \frac y h \, \sum_{i=2}^{n}\int_{(i-1)h}^{ih} \dd x \, F(x,y) \Bigl(1-\frac{ih-x}{h}\Bigr) 
		\\
		&\quad + 
		\sum_{i,j=2}^{n}\int_{(i-1)h}^{ih}\dd x\int_{(j-1)h}^{jh}\dd y \, F(x,y) \Bigl(1-\frac{ih-x}{h}\Bigr)\Bigl(1-\frac{jh-y}{h}\Bigr)
		\\
		&\le \int_0^h \dd x \int_0^{\Delta} \dd y \, F(x,y)
		+ \int_0^h \dd y \int_0^{\Delta} \dd x \, F(x,y)
		\\
		&\quad +
		\frac{h^{2}}4  \sum_{i,j=2}^{n} F\left((i-1)h,(j-1)h\right) =: \overline \delta(n),
	\end{aligned}
	\]
	where we have used again the change of variables $z_{1}=\frac{ih-x}{h}$, $z_{2}=\frac{jh-y}{h}$, the inequality \\ $F\left((i-z_{1})h,(j-z_{2})h\right)$ $\le F\left((i-1) h,(j-1) h\right)$ and the identity  $\int_{0}^{1}\int_{0}^{1}(1-z_{1})(1-z_{2})\dd z_{1}\dd z_{2}=\frac{1}{4}$ in the last step.
	Since, on the one hand,
	\[
	\int_0^h \dd x \int_0^{\Delta} \dd y \, F(x,y)
	+ \int_0^h \dd y \int_0^{\Delta} \dd x \, F(x,y) \xrightarrow[h\to 0]{} 0,
	\]
	and, on the other hand,
	\[
	\begin{aligned}
		\frac{h^2}4 \sum_{i,j=2}^{n} F\left((i-1)h,(j-1)h\right)
		&= \frac{h^2}4 \sum_{i,j=1}^{n-1} F\left(i h,j h\right)
		\\
		&\le \frac{h^2}4 \sum_{i,j=1}^{n} F\left(i h,j h\right)
		\xrightarrow[h\to 0]{} \frac 14 \int_{0}^{\Delta}\int_{0}^{\Delta}F(x,y)\dd x\dd y = \frac 14 \lambda_\Delta,
	\end{aligned}
	\]
	we obtain $\overline \delta(n) \to \frac 14 \lambda_\Delta$ as $n \to \infty$, as claimed.
\end{proof}

\section{Acknowledgments and Declaration of Interest}

We thank our colleagues Emmanuel Gobet and Gilles Pag\`{e}s for stimulating discussions on multilevel Monte Carlo methods.
This research benefited from the support of the \emph{Chaire Risques Financiers}, led by Ecole Polytechnique and \emph{La Fondation du Risque} and sponsored by Soci\'{e}t\'{e} G\'{e}n\'{e}rale, and of the \emph{Chaire Stress Test, RISK Management and Financial Steering}, led by Ecole Polytechnique and \emph{Fondation de l'Ecole Polytechnique} and sponsored by BNP Paribas.

\bibliographystyle{plainnat}
\bibliography{biblio}

\begin{thebibliography}{26}
\providecommand{\natexlab}[1]{#1}
\providecommand{\url}[1]{\texttt{#1}}
\expandafter\ifx\csname urlstyle\endcsname\relax
  \providecommand{\doi}[1]{doi: #1}\else
  \providecommand{\doi}{doi: \begingroup \urlstyle{rm}\Url}\fi

\bibitem[Alos et~al.(2018)Alos, Garc\'ia-Lorite, and Muguruza]{Alos2018VIXskew}
E.~Alos, D.~Garc\'ia-Lorite, and A.~Muguruza.
\newblock {On smile properties of volatility derivatives and exotic products: understanding the VIX skew}.
\newblock https://arxiv.org/abs/1808.03610, 2018.

\bibitem[Alos et~al.(2007)Alos, Le{\'o}n, and Vives]{alos2007short}
Elisa Alos, Jorge~A Le{\'o}n, and Josep Vives.
\newblock On the short-time behavior of the implied volatility for jump-diffusion models with stochastic volatility.
\newblock \emph{Finance and Stochastics}, 11\penalty0 (4):\penalty0 571--589, 2007.

\bibitem[Asmussen and Glynn(2007)]{asmussen2007stochastic}
S{\o}ren Asmussen and Peter~W Glynn.
\newblock \emph{Stochastic simulation: algorithms and analysis}, volume~57.
\newblock Springer Science \& Business Media, 2007.

\bibitem[Bayer et~al.(2016)Bayer, Friz, and Gatheral]{bayer2016pricing}
Christian Bayer, Peter Friz, and Jim Gatheral.
\newblock Pricing under rough volatility.
\newblock \emph{Quantitative Finance}, 16\penalty0 (6):\penalty0 887--904, 2016.
\newblock ISSN 1469-7688.
\newblock \doi{10.1080/14697688.2015.1099717}.
\newblock URL \url{https://doi.org/10.1080/14697688.2015.1099717}.

\bibitem[Bayer et~al.(2021)Bayer, Hall, and Tempone]{bayer2020weak}
Christian Bayer, Eric~Joseph Hall, and Raul Tempone.
\newblock Weak error rates for option pricing under linear rough volatility.
\newblock \url{https://arxiv.org/abs/2009.01219}, 2021.

\bibitem[Ben~Alaya and Kebaier(2014)]{alaya2014multilevel}
Mohamed Ben~Alaya and Ahmed Kebaier.
\newblock {Multilevel Monte Carlo for Asian options and limit theorems}.
\newblock \emph{{M}onte {C}arlo Methods Appl.}, 20\penalty0 (3):\penalty0 181--194, 2014.
\newblock ISSN 0929-9629.
\newblock \doi{10.1515/mcma-2013-0025}.
\newblock URL \url{https://doi.org/10.1515/mcma-2013-0025}.

\bibitem[Bergomi(2016)]{Bergomi2016}
L.~Bergomi.
\newblock \emph{{Stochastic Volatility Modeling}}.
\newblock Chapman and Hall/CRC, 2016.

\bibitem[Bergomi(2005)]{bergomi2005smile}
Lorenzo Bergomi.
\newblock Smile dynamics {II}.
\newblock \emph{Risk}, pages 67--73, 2005.

\bibitem[Bonesini et~al.(2021)Bonesini, Callegaro, and Jacquier]{Bonesini2021functional}
O.~Bonesini, G.~Callegaro, and A.~Jacquier.
\newblock {Functional quantization of rough volatility and applications to the VIX}.
\newblock https://arxiv.org/abs/2104.04233, 2021.

\bibitem[Bourgey(2020)]{bourgey2020stochastic}
Florian Bourgey.
\newblock \emph{Stochastic approximations for financial risk computations}.
\newblock PhD thesis, Institut Polytechnique de Paris, 2020.

\bibitem[{Chicago Board Options Exchange}(2009)]{exchange2009cboe}
{Chicago Board Options Exchange}.
\newblock {The CBOE Volatility Index-VIX}.
\newblock \emph{www.cboe.com/micro/vix/vixwhite.pdf}, pages 1--23, 2009.

\bibitem[Dubois and Leli{\`e}vre(2004)]{dubois2004efficient}
Fran{\c{c}}ois Dubois and Tony Leli{\`e}vre.
\newblock {Efficient pricing of Asian options by the PDE approach}.
\newblock \emph{Journal of Computational Finance}, 8\penalty0 (2):\penalty0 55--64, 2004.

\bibitem[Dupire(1993)]{dupire1992arbitrage}
Bruno Dupire.
\newblock Arbitrage pricing with stochastic volatility.
\newblock https://cims.nyu.edu/~essid/ctf/stochvol.pdf, 1993.

\bibitem[Fukasawa(2017)]{fukasawa:2017}
Masaaki Fukasawa.
\newblock Short-time at-the-money skew and rough fractional volatility.
\newblock \emph{Quant. Finance}, 17\penalty0 (2):\penalty0 189--198, 2017.
\newblock ISSN 1469-7688.
\newblock \doi{10.1080/14697688.2016.1197410}.
\newblock URL \url{https://doi.org/10.1080/14697688.2016.1197410}.

\bibitem[Giles(2008)]{giles_2008}
Michael~B Giles.
\newblock Multilevel {M}onte {C}arlo path simulation.
\newblock \emph{Oper. Res.}, 56\penalty0 (3):\penalty0 607--617, 2008.
\newblock ISSN 0030-364X.
\newblock \doi{10.1287/opre.1070.0496}.
\newblock URL \url{https://doi.org/10.1287/opre.1070.0496}.

\bibitem[Giles(2015)]{giles_2015}
Michael~B Giles.
\newblock Multilevel {M}onte {C}arlo methods.
\newblock \emph{Acta Numer.}, 24:\penalty0 259--328, 2015.
\newblock ISSN 0962-4929.
\newblock \doi{10.1017/S096249291500001X}.
\newblock URL \url{https://doi.org/10.1017/S096249291500001X}.

\bibitem[Guyon(2020)]{guyon2020vix}
Julien Guyon.
\newblock {The VIX Future in Bergomi Models: Analytic Expansions and Joint Calibration with S\&P 500 Skew}.
\newblock https://ssrn.com/abstract=3720315, 2020.

\bibitem[Horvath et~al.(2020)Horvath, Jacquier, and Tankov]{horvath2018volatility}
Blanka Horvath, Antoine Jacquier, and Peter Tankov.
\newblock Volatility options in rough volatility models.
\newblock \emph{SIAM J. Financial Math.}, 11\penalty0 (2):\penalty0 437--469, 2020.
\newblock \doi{10.1137/18M1169242}.
\newblock URL \url{https://doi.org/10.1137/18M1169242}.

\bibitem[Jacquier et~al.(2018)Jacquier, Martini, and Muguruza]{Jacquier2018VIXFutures}
A.~Jacquier, C.~Martini, and A.~Muguruza.
\newblock {On VIX futures in the rough Bergomi model}.
\newblock \emph{Quantitative Finance}, 18\penalty0 (1):\penalty0 45--61, 2018.

\bibitem[Kemna and Vorst(1990)]{kemna1990pricing}
Angelien G~Z Kemna and Antonius C~F Vorst.
\newblock A pricing method for options based on average asset values.
\newblock \emph{Journal of Banking \& Finance}, 14\penalty0 (1):\penalty0 113--129, 1990.

\bibitem[Lapeyre and Temam(2001)]{lapeyre2001competitive}
Bernard Lapeyre and Emmanuel Temam.
\newblock {Competitive Monte Carlo methods for the pricing of Asian options}.
\newblock \emph{{Journal of Computational Finance}}, 5\penalty0 (1):\penalty0 39--58, 2001.

\bibitem[Neuenkirch and Shalaiko(2016)]{neuenkirch2016order}
Andreas Neuenkirch and Taras Shalaiko.
\newblock The order barrier for strong approximation of rough volatility models.
\newblock https://arxiv.org/abs/1606.03854, 2016.

\bibitem[Olver et~al.(2010)Olver, Lozier, Boisvert, and Clark]{olver2010nist}
Frank W~J Olver, Daniel~W Lozier, Ronald~F Boisvert, and Charles~W Clark, editors.
\newblock \emph{{NIST handbook of mathematical functions}}.
\newblock U.S. Department of Commerce, National Institute of Standards and Technology, Washington, DC; Cambridge University Press, Cambridge, 2010.
\newblock ISBN 978-0-521-14063-8.

\bibitem[Rogers and Shi(1995)]{rogers1995value}
L~Chris~G Rogers and Zo~Shi.
\newblock The value of an asian option.
\newblock \emph{Journal of Applied Probability}, pages 1077--1088, 1995.

\bibitem[Virtanen et~al.(2020)Virtanen, Gommers, Oliphant, Haberland, Reddy, Cournapeau, Burovski, Peterson, Weckesser, Bright, et~al.]{virtanen2020scipy}
Pauli Virtanen, Ralf Gommers, Travis~E Oliphant, Matt Haberland, Tyler Reddy, David Cournapeau, Evgeni Burovski, Pearu Peterson, Warren Weckesser, Jonathan Bright, et~al.
\newblock {SciPy 1.0: fundamental algorithms for scientific computing in Python}.
\newblock \emph{Nature methods}, 17\penalty0 (3):\penalty0 261--272, 2020.

\bibitem[Zhang(2001)]{zhang2001semi}
J~Zhang.
\newblock A semi-analytical method for pricing and hedging continuously sampled arithmetic average rate options.
\newblock \emph{Journal of Computational Finance}, 5\penalty0 (1):\penalty0 59--80, 2001.

\end{thebibliography}

\newpage

\listoffigures
\listoftables	

\medskip

%Our submission discusses discretization and multilevel Monte Carlo simulation of the VIX index in the rough Bergomi model. Our main contributions are:
%\begin{itemize}
%	\item Bounds on the strong error for rectangle and trapezoidal schemes, exact asymptotics of the $L^2$ strong error for the rectangle scheme.
%	\item The asymptotically optimal cost $\cO(\ve^{-2})$ in the valuation of VIX options is achieved combining a multilevel estimator with the trapezoidal scheme.
%	\item We provide numerical evidence of the efficiency of the multilevel method for the pricing of VIX options.
%\end{itemize}

\end{document}